\documentclass[12pt,fleqn]{article}

\input{paper_preamble.sty}

\newcommand{\edq}[1]{{\color{Mgreen}{#1}}}

\title{\bf Bianchi cosmologies in a Thurston-based \\ theory of gravity}

\author[1,2]{
    {Quentin} 
    {Vigneron}\footnote{\href{mailto:quentin.vigneron@ens-lyon.fr}{quentin.vigneron@ens-lyon.fr}}}
\author[2]{
    {Hamed} 
    {Barzegar}\footnote{\href{mailto:hamed.barzegar@ens-lyon.fr}{hamed.barzegar@ens-lyon.fr}}}

\affil[1]{\small\it{Institute of Astronomy, Faculty of Physics, Astronomy and Informatics}, {Nicolaus Copernicus University}, {{Grudzi{\k{a}}dzka 5}, {Toru\'n}, {87-100} {Poland}}}

\affil[2]{\small\it{ENS de Lyon, CRAL UMR5574, Universite Claude Bernard Lyon 1, CNRS, Lyon 69007, France}}
\date{\vspace{-.5cm}(\today)}

\begin{document}

\maketitle

\begin{abstract}

The strong interplay between Bianchi--Kantowski--Sachs (BKS) spacetimes and Thurston geometries motivates the exploration of the role of topology in our understanding of gravity. As such, we study non-tilted BKS solutions of a theory of gravity that explicitly depends on Thurston geometries.
We show that shear-free solutions with perfect fluid, as well as static vacuum solutions, exist for all topologies. Moreover, we prove that, aside from non-rotationally-symmetric Bianchi II models, all BKS metrics isotropize in the presence of a positive cosmological constant, and that recollapse is never possible when the weak energy condition is satisfied. This contrasts with General Relativity (GR), where these two properties fail for Bianchi~IX and KS metrics. No additional parameters compared to GR are required for these results. We discuss, in particular, how this framework might allow for simple inflationary models in any~topology.

\end{abstract}

\newpage
\setcounter{tocdepth}{2}

\setlength{\cftbeforesecskip}{5pt}
\tableofcontents
\newpage

\section{Introduction}

Locally spatially homogeneous (LSH) spacetimes are interesting classes of solutions to theories of gravity because of their rich structures and, to some extent, their usefulness in describing the Universe as a whole (cf., e.g., \cite{1975_RyanShepley_homogeneousRC, 1980_KramerStephani_exact, 1997_Wainwright_et_al_BOOK}) or local structures in it (see, e.g., \cite{2021_Giani_BIX, 2024_Malkiewicz_Ostrowski_BIX}). In fact, the first of such solutions (i.e., vacuum Bianchi type I solution) is almost as old as the birth of the General Relativity (GR) by the work of Kasner \cite{1921_Kasner}, which was not motivated by group theoretical considerations. However, Bianchi models (as the major class of LSH spacetimes) based on Bianchi's classification in his seminal work \cite{1898_Bianchi, 2001_Bianchi_Jantzen, 2001_Jantzen_editorsnote}, were introduced first in the famous solution by G\"odel in 1949 that he gave as a present on the occasion of the 70th birthday of Einstein (see, e.g., \cite{2009_Rindler_Goedel}). G\"odel's groundbreaking work included Bianchi types~III and VIII  (cf.~\cite{Ozsvath_1970_ClassIIadnIII, 1982_Jantzen_quaternions, 2001_Jantzen_editorsnote}), and was followed by the presentation of his results (though without proof) on Bianchi type IX in 1950 \cite{1952_Godel_rotating, 2000_Godel_reprint1952}.\footnote{G\"odel remained actively engaged with ongoing developments in Bianchi cosmology until the final years of his life; see \cite{2001_Jantzen_editorsnote} and \cite{2009_Rindler_Goedel} for a summary of his remarkable role in GR and relativistic cosmology.} 
Nevertheless, it was Taub (1951) who systematically used Bianchi's classification for LSH spacetimes, i.e., where the spatial sections admit a simply transitive $3$-dimensional isometry group \cite{1951_Taub_empty}---as G\"odel investigated homogeneous spacetimes---which was revived later by Heckmann and Sch\"ucking in the late 1950's \cite{1962_HeckmannSchuecking_RelativisticCosmology}. Finally, Kantowski and Sachs (1966) investigated LSH spacetimes  with no simply transitive subgroup, known as Kantowski--Sachs solutions~\cite{1966_KantowskiSach_SH}, thereby completing the classification of LSH cosmological spacetimes.\footnote{Although Bianchi already considered this case in \cite[\S9]{1898_Bianchi} (see \cite[Section~9]{2001_Bianchi_Jantzen} for the English translation), he ``curiously'' omitted the discussion of $3$-geometries that possess a higher degree of symmetry but lack a simply transitive $3$-dimensional subgroup when summarizing the complete classification in \cite[\S38]{1898_Bianchi} (\cite[Section~38]{2001_Bianchi_Jantzen}); see also \cite{2001_Jantzen_editorsnote}. Hence, this justifies referring to LSH solutions as simply Bianchi cosmologies, a designation used in the title of our paper.}
Since then, LSH cosmological models have attracted a lot of interest, for example analysis of: 
the dynamics towards the singularity, in particular the mixmaster universe and the BKL conjecture \cite{1969_Misner, 1970_Belinskij_et_al} (see \cite{2025_Ringstroem_BKL} and references therein for recent developments on this topic), the vacuum regime (e.g., \cite{2013_Ringstroem_toplogy, 2013_Ringstroem_Cauchy}), the late time dynamics and the question of the isotropization of anisotropic initial conditions (e.g., \cite{1973_Collins_et_al, 1976_Barrow, 1983_Wald, 1997_Wainwright_et_al_BOOK, 2003_Coley_DynamicalSystem}), the dependence on the matter content (e.g., \cite{1970_Maccallum_et_al_b, 1997_Wainwright_et_al_BOOK, 2000_Coley_et_al, 2003_Coley_DynamicalSystem, 2008_Rendall_PDEinGR, 2011_Calogero_et_al, 2014_Fadragas_et_al, 2021_Barzegar, 2018_NormannHervik_pfields}) and its tilt (e.g., \cite{1973_King_et_al, 1999_Goliath_et_al}), and the description of first order perturbations (e.g., \cite{2003_Tanimoto, 2003_Tanimoto_et_al, 2007_Pereira_et_al, 2012_Pereira_et_al, 2015_Pereira_et_al, 2016_Pereira_et_al, 2017_Franco_et_al}). Additionally, the full classification of closed $3$-manifolds (cf.~\cite{1982_Thurston,2002_Perelman,2003_Perelman_a,2003_Perelman_b}) raised the question of determining how spatial topology constrains LSH cosmological models, and in particular how it constrains the spatial geometry (e.g., \cite{1985_Fagundes, 1994_Koike_et_al,1998_Kodama,2002_Kodama}), the matter content (e.g., \cite{2002_Kodama, 2001_Barrow_et_al_a}) and their dynamics (e.g., \cite{1993_Fujiwara_et_al, 1997_Tanimoto_et_al, 1998_Kodama, 2001_Barrow_et_al_a, 2001_Barrow_et_al_b, 2002_Kodama,2025_Smith_et_al}). 
A strong focus has also been made on considering modified gravity theories in the context of LSH cosmological models (e.g., \cite{1995_Mimoso_et_al, 2007_Goheer_et_al}). They may offer a better framework than GR with regards to, e.g., resolving cosmological singularities \cite{2019_de_Cesare_et_al} or for constructing inflation \cite{2024_Vigneron_et_al_b}.
Last but not least, Bianchi models serve as minisuperspace models in quantum gravity approaches: they reduce the infinite-dimensional superspace of all $3$-geometries to a finite number of degrees of freedom, making quantization tractable (cf., e.g., \cite{1991_Ashtekar_et_al_topology, 1992_Ashtekar_et_al_symmetries, 1994_Koike_et_al, 2010_Barbero_et_al_quantiziation, 2022_Geiller_et_al_dynamicalsymmetries}).

From an observational point of view, Bianchi cosmologies offer a crucial framework for testing the fundamental assumption of isotropy in the standard cosmological model that assumes GR. While early analysis using the WMAP data had a preference for a Bianchi VII$_h$ model exhibiting global vorticity and shear \cite{2006_Jaffe_et_al_BVIIhWMAP}, the inclusion of dark energy \cite{2006_Jaffe_et_al_BVIIhruledout} and the later Planck observations \cite{2013_PlanckData_topology, 2015_PlanckData_topology} found no evidence for a significant anisotropic component. This placed stringent limits on such models, and at the same time on the topology of the Universe.\footnote{These analyses were however not exhaustive in different topological degrees of freedom that can be considered \cite{2024_Akrami_et_al}. In this regard, a more exhaustive analysis on how much the Planck data constrain topology is currently being conducted by the COMPACT collaboration \cite{2023_Petersen_et_al_COMPACT, 2024_Mihaylov_et_al_COMPACT_ERRATUM, 2024_Eskilt_et_al_COMPACT, 2024_Tamosiunas_et_al_COMPACT, 2024_Samandar_et_al_COMPACT, 2025_Saha_et_al_COMPACT, 2025_Samandar_et_al_COMPACT, 2025_Copi_et_al_COMPACT}.} These findings are in agreement with the standard cosmological model in which inflation plays a fundamental role in driving any initial global anisotropy to zero at late times. This result finds its root in Wald's theorem \cite{1983_Wald}, stating that negative curvature models become isotropic asymptotically during inflation. While most LSH models fulfill that condition, Bianchi IX ($\mS^3$ topologies) and Kantowski--Sachs models ($\RSS$ topologies) do not, and therefore require additional fine-tuning on their initial conditions to explain the late-time isotropy and avoid an early recollapse (cf., e.g., \cite{1987_Barrow, 1988_Barrow, 1988_Gotz, 1988_Ponce, 1993_Moniz}). In the absence of a physical mechanism excluding these topologies for our Universe, having a non-fine-tuned mechanism that prevents their recollapse and in the same time leads to their isotropization is needed. This motivates, in particular, the study of anisotropic models in modified gravity theories (see, e.g., \cite{1999_Cervantes_Cota_et_al, 2005_Miritzis}).

Thus, there is a strong interplay between LSH models and spatial topology, with motivations ranging from early universe physics to late time isotropization, and in this context, modified gravity theories can offer new perspectives. The goal of this paper is to study this interplay within a recently proposed {\it parameter-free} modified gravity theory \cite{2024_Vigneron}. In this framework, named {\it topo-GR}, a reference (topological) spacetime Ricci curvature $\stTbarRic$ depending on the spatial topology is added to the Einstein equation. This direct link between the field equation of topo-GR and topology naturally suggests strong departures with GR regarding LSH solutions. Indeed, already for homogeneous-isotropic metrics—hence restricting the spatial topology to be Euclidean, spherical or hyperbolic—it was shown in \cite{2023_Vigneron_et_al_b} that the expansion of homogeneous-isotropic models in topo-GR does not depend on the spatial curvature parameter, hence contrasting with GR. This result found a first application in \cite{2024_Vigneron_et_al_b} where it was shown that it allows for a canonical quantization of inflation and a Bunch--Davies prescription in spherical and hyperbolic topologies, something not possible in GR.

Our goal is to provide the systems of equations of LSH solutions in topo-GR with non-tilted perfect fluids in each of the (Thurston) families of topologies. We do not aim at providing a full dynamical analysis of these equations. Rather we will present specific solutions, such as static vacuum solutions and shear-free anisotropic solutions. Furthermore, we will show that, except for a subcase of Bianchi II metrics, no recollapse is possible and late time isotropization is ensured by a positive cosmological constant, hence generalizing the result of Wald \cite{1983_Wald} for topo-GR to the $\mS^3$ and $\RSS$ topologies.

The paper is organized as follows. Before proceeding to
\Cref{sec_Max_Min} that reviews the Thurston classification and its correspondence to Bianchi and Kantowski--Sachs metrics, we begin by establishing the notations, definitions, and conventions used throughout this paper in \Cref{sec_notation}. 
\Cref{sec_topo_GR} recalls the fundamentals of topo-GR and introduces its field equations. 
Section~\ref{sec_topo_GR_3_1} derives the $3 + 1$ equations of topo-GR and their homogeneous version for a generic LSH solution. In particular, we derive a Wald-like theorem for isotropization in the presence of a cosmological constant. 
Sections~\ref{sec_all_models} computes the precise field equations for each type of LSH solutions. 
Finally, \Cref{sec_discussion_Nil} discusses a peculiar result for $\Nils$-topologies, and \Cref{sec_Conclusion} concludes our study.

\subsection{Notations, definitions, and conventions}
\label{sec_notation}

Throughout this paper, we will use the following notations. We consider a spacetime 4-manifold $(\CM, \T g)$  where $\CM \cong \mR\times\Sigma$ with $\Sigma$ a closed Riemannian $3$-manifold. Bold letters are used to represent tensors in a coordinate-free notation. Indices running from 0 to 3 will be represented by Greek letters and those running from 1 to 3 by Roman letters. Local coordinate functions are denoted by $(x^\mu) = (t,x^i) \equiv (t, \T x)$.
We occasionally use the musical symbols $\sharp$ and $\flat$---for the \emph{musical isomorphism}---to denote the metric-dual $1$-form and metric-dual vector of a given vector and $1$-form, respectively, for a given metric. 
Angle brackets represent the (projected) symmetric tracefree part of a given tensor $\T A$ with components $A_{ij}$ on a Riemannian $3$-manifold $(\Sigma,\T h)$, i.e., ${A_{\langle ij \rangle} \coloneqq A_{(ij)} - h_{ij}} h^{k\ell} A_{k\ell} / 3$. For a metric $\T h$ we denote by $\T\nabla[\T h]$ and $\TRic[\T h]$ its Levi-Civita connection and Ricci tensor, respectively, denoted by $\T\nabla$ and $\TRic$, whenever no ambiguity arises. The corresponding objects associated with the metric $\Tbarh$ will be designated by $\Tbarnabla$ and $\sTbarRic$. 
To distinguish spacetime objects from spatial ones we occasionally use an upper-left index $4$ (e.g., $\stTnabla$). The semidirect product is denoted by  $\rtimes$.
For the ease of reference, we summarize all the groups used in this work in \Cref{tab_def_Var}. Finally, we set the speed of light at $c = 1$.

A foliation will be called an  \emph{$\T n$-foliation} whenever the $1$-form $\T n$ is orthogonal to the corresponding foliation of $\CM$, and it is called a \emph{$\Sigma$-foliation} if the leaves of the foliation of $\CM$ are diffeomorphic to a hypersurface $\Sigma$.

For any manifold $\Sigma$, there is a unique universal covering space denoted by $\tilde\Sigma$. Any geometric quantity defined on $\Sigma$ can be lifted to a geometric quantity on $\tilde\Sigma$, but the reverse does not hold necessarily. In this work, unless an ambiguity is possible, we will denote equivalently the quantity defined on $\Sigma$ and the one lifted on~$\tilde\Sigma$.

Given a manifold $\Sigma$ and a tensor field $\T F$ on $\Sigma$, we define
\begin{equation}
    \Sym(\Sigma,\T F)
        \coloneqq \left\{\phi \in \Diff(\Sigma) \ | \ \phi^* \T F = \T F \right\}, \quad
    \sym(\Sigma,\T F)
		\coloneqq \left\{\T\xi \in \mathfrak{X}(\Sigma) \ | \ \Lie{\T \xi} \T F = 0 \right\},
\end{equation}
where $\mathfrak{X}(\Sigma)$ is the set of all vector fields on $\Sigma$. Given a torsion-free connection $\T \nabla$ on $\Sigma$, we define
\begin{equation}
    \sym(\Sigma,\T\nabla)
		\coloneqq \left\{\T\xi \in \mathfrak{X}(\Sigma) \ | \ \left(\Lie{\T \xi} \T \nabla\right)^\mu{}_{\alpha\beta} = \nabla_\alpha\nabla_\beta \xi^\mu + \Riem[\T\nabla]^\mu{}_{\beta\alpha\gamma}\,\xi^\gamma = 0 \right\}.
\end{equation}
For a tensor field $\T F$, we say that a (local) vector field $\T\xi$ is a (local) \emph{$\T F$-colineation} if $\Lie{\T\xi} \T F= 0$. If $\T F$ is a metric tensor, we will instead say that $\T\xi$ is a \emph{Killing vector field}.

Finally, for a (semi-)Riemannian manifold $(\Sigma,\T h)$, we define
\begin{equation}
    \begin{alignedat}{4}
	\Isom(\Sigma,\T h) &\coloneqq \Sym(\Sigma, \T h)\,, \quad&
        \isom(\Sigma,\T h) &\coloneqq \sym(\Sigma, \T h)\,, \\
	\Aff(\Sigma,\T h) &\coloneqq \Sym(\Sigma, \T\nabla[\T h])\,, \quad&
        \aff(\Sigma,\T h) &\coloneqq \sym(\Sigma, \T\nabla[\T h])\,, \\
	\ColRic(\Sigma,\T h) &\coloneqq \Sym(\Sigma, \TRic[\T h])\,, \quad&
        \colRic(\Sigma,\T h) &\coloneqq \sym(\Sigma, \TRic[\T h])\,,
    \end{alignedat}
\end{equation}
for which the following inclusions hold (cf.~\cite{1969_Katzin_et_al})\footnote{While $\Isom(\Sigma,\T h)$ and $\Aff(\Sigma,\T h)$ are always (finite dimensional) Lie groups, $\ColRic(\Sigma,\T h)$ can be infinite dimensional.}
\begin{equation}
    \Isom(\Sigma,\T h) \subseteq \Aff(\Sigma,\T h) \subseteq \ColRic(\Sigma,\T h)\,, \\
    \isom(\Sigma,\T h) \subseteq \aff(\Sigma,\T h) \subseteq \colRic(\Sigma,\T h)\,.
\end{equation}

\begin{table}[h]
\centering\small
\caption{\small Thurston geometries and their correspondence to Bianchi--Kantowski--Sachs metrics (based on \cite{1985_Fagundes, 1993_Fujiwara_et_al, 1994_Koike_et_al, 1997_Tanimoto_et_al, 1998_Kodama, 2002_Kodama}): $(i)$ $\Gmax$ is the group of the maximal geometry; $(ii)$ $\Gmin$ is the group of a minimal subgeometry of $(X,\Gmax)$; $(iii)$ $\Isom(\tilde\Sigma,\T h_{\Gmin})_{\rm general}$ is the local isometry group of a \emph{general} minimal metric $\T h$ defined on a closed manifold $\Sigma$.
\label{tab_Gmax_Gmin}}
\renewcommand{\arraystretch}{2.2}
\renewcommand{\tabcolsep}{10pt}
\newcommand{\vspacetable}[1]{\multicolumn{2}{c}{} \vspace{#1}\\}
\begin{tabular}{llll}

\toprule
\vspacetable{-36pt}
\makecell[l]{Geometry}
& $\Gmax$
& $\Isom(\tilde\Sigma,\T h_{\Gmin})_{\rm general}$
& $\Gmin$
\\
\midrule

$\mE^3$
& $\IO(3)$
& $\begin{dcases} \Gmax \\ \BVIIn \rtimes {\rm D}_2 \end{dcases}$
& $\begin{dcases} \BI = \mR^3  \\ \BVIIn \end{dcases}$
\\

$\mS^3$
& $\rmO(4)$
& ${\rm SU}(2)\rtimes \rmD_2$
& $\BIX = {\rm SU}(2)$
\\

$\mH^3$
& $\rmO_+(1,3)$
& $\Gmax$
& $\begin{dcases} \BV \\ \BVIIh \end{dcases}$
\\

$\mR\times\mH^2$
& $\rmO_+(2,1) \times \IO(1)$
& $\Gmax$
& $\BIII$
\\

$\Nils$
& $\Nil \rtimes \rmO(2)$
& $\Gmax$ 
& $\BII = \Nil$
\\

$\Sols$
& $\Sol \rtimes \rmD_4$
& $\Sol \rtimes \rmD_2$
& $\BVIn = \Sol$ 
\\

$\SLRRs$
& ${\SLRR} \rtimes \rmO(2)$
& $\begin{dcases} \SLRR \rtimes \rmD_2  \\ \Gmax \end{dcases}$
& $\begin{dcases} \BVIII = \SLRR \\ \BIII \end{dcases}$
\\

$\mR\times\mS^2$
& $\rmO(3) \times \IO(1)$
& $\Gmax$
& $\Gmax$
\\
\bottomrule

\end{tabular}
\end{table}

\section{Maximal and minimal geometries with their metrics} \label{sec_Max_Min}

\subsection{Maximal geometries}\label{sec_Thurston}

In this section, we briefly introduce the Thurston classification, and refer the reader for more technical details to, e.g., \cite{1982_Thurston, 1997_Thurston_threegeometry, 1983_Scott_threegeometries, 1994_Koike_et_al, 1997_Tanimoto_et_al, 1998_Kodama, 2002_Kodama}.

A \emph{geometry} is defined by a pair $(X, G)$, where $X$ is a manifold and $G$ is a group acting transitively on $X$ such that the point stabilizer (isotropy subgroup) is compact. These properties guarantee the existence of a complete $G$-invariant Riemannian metric on $X$.\footnote{In this sense, a geometry is an equivalence class of homogeneous Riemannian manifolds.} A geometry $(X, {G}')$ is a \emph{subgeometry} of  $(X, G)$ if ${G}'$ is a subgroup of $G$. We call a geometry \emph{minimal} if it does not have a proper subgeometry, and denote it by $(X,\Gmin$). We call a geometry \emph{maximal} if it is not a proper subgeometry of any geometry, and denote it by $(X,\Gmax$). For such maximal geometries, no larger group includes $G$ as a proper subgroup while maintaining the transitive action and compact stabilizer condition. Such maximal geometries are often termed \emph{Thurston}; indeed, they were classified by Thurston which we summarize in the following theorem.\footnote{The classification of maximal Lorentzian geometries admitting compact quotients has also been done, see, e.g., \cite{2007_Dumitrescu_et_al}. There are four maximal Lorentzian geometries on top of the ones that are also maximal Riemannian geometries. These additional geometries do not have compact point stabilizer. Moreover, a maximal Lorentzian metric is not necessarily complete, contrary to Riemannian maximal metrics.}
\begin{theorem}[Thurston, 1982 \cite{1982_Thurston, 1997_Thurston_threegeometry}]
The eight maximal, simply connected $3$-dimen{-}sional geometries $(X, \Gmax)$ admitting compact quotient manifolds are as follows:
\begin{enumerate}
    \item the $\mE^3$-geometry (or Euclidean geometry): $X = \mR^3$ and $\Gmax = \IO(3) \eqcolon \GE$,

    \item the $\mS^3$-geometry (or spherical geometry): $X = \mS^3$ and $\Gmax = \rmO(4) \eqcolon \GS$,

    \item the $\mH^3$-geometry (or hyperbolic geometry): $X = \mR^3$ and $\Gmax = \rmO_+(3,1) \eqcolon \GH$,

    \item the $\RSS$-geometry: $X =\RSS$ and $\Gmax = \rmO(3) \times \IO(1) \eqcolon \GRSS$,

    \item the $\RHH$-geometry: $X = \mR^3$ and $\Gmax = \rmO_+(2,1) \times \IO(1) \eqcolon \GRHH$,
    
    \item the $\Nils$-geometry: $X = \mR^3$ and $\Gmax = \Nil \rtimes \rmO(2) \eqcolon \GNil$,
    
    \item the $\Sols$-geometry: $X = \mR^3$ and $\Gmax = \Sol \rtimes \rmD_4 \eqcolon \GSol$,

    \item the $\SLRRs$-geometry: $X = \mR^3$ and $\Gmax = {\SLRR} \rtimes \rmO(2) \eqcolon \GSLRR$.
\end{enumerate}
\end{theorem}

A given group $\Gmax$ defines a family (or class) of $3$-dimensional topological spaces. For this reason we will sometimes refer to ``the class of topology'' of a closed $3$-manifold. A precise topological space within each class is characterized by a discrete subgroup $\Gamma$ of~$\Gmax$.

\begin{definition}[Geometric manifold]
A manifold $\Sigma$ is a  \emph{geometric $3$-manifold} if it is a closed $3$-manifold and $\Sigma = \tilde\Sigma/\Gamma$ where $\Gamma$ is a discrete subgroup of one of the maximal geometries, and  $\tilde\Sigma$ is the universal covering space of $\Sigma$. Equivalently, we say that $\Sigma$ is modeled on a Thurston geometry.
\end{definition}
\begin{theorem}[Thurston--Hamilton--Perelman, 2003 \cite{1982_Thurston, 1982_Hamilton_Ricci, 2002_Perelman, 2003_Perelman_a}]\label{thm_THP}
Every closed $3$-manifold is a connected sum of geometric $3$-manifolds.
\end{theorem}
\begin{remark}
\Cref{thm_THP} states that the knowledge of every geometric $3$-manifold, and equivalently the knowledge of every discrete subgroup $\Gamma$ (acting freely on $X$) of the maximal geometries is sufficient to construct every closed $3$-manifold. However, geometric $3$-manifolds are the only closed $3$-manifolds admitting locally homogeneous Riemannian metrics. Otherwise, when a non-trivial connected sum is considered, no such metric can be defined (the only exception is $\mR\mP^3 \# \mR\mP^3$ which can be modeled on $\RSS$, see \cite{1974_Tollefson_stwoxsone, 1983_Scott_threegeometries}). For this reason, in this paper, we will only consider manifolds modeled on a maximal geometry. Moreover, we will not be interested in the precise topology, given by the discrete subgroup $\Gamma \subset \Gmax$, but rather we will be interested only in the maximal geometry on which the manifold is modeled.
\end{remark}
\begin{definition}[Maximal metric]
Given a closed  $3$-manifold $\Sigma$ modeled on a $(X,\Gmax)$ geometry, a Riemannian metric $\T h$ on $\Sigma$ is said to be \emph{maximal}, or \emph{Thurston}, if the isometry group of its lift on $\tilde\Sigma$ is $\Gmax$, i.e., $\Isom(\tilde\Sigma,\T h) = \Gmax$. 
\end{definition}
Since $\Gmax$ acts transitively on $\tilde\Sigma$, maximal metrics are locally homogeneous metrics on~$\Sigma$.

\subsection{Minimal geometries}\label{sec_Minimal_geometries}

Locally homogeneous metrics that are not necessarily maximal are defined to be locally invariant by a minimal geometry. It was shown in \cite{1985_Fagundes, 1994_Koike_et_al, 1997_Tanimoto_et_al}, summarized in the following theorem, that, apart from the $\RSS$-geometry, all the minimal subgeometries of the maximal geometries are described by simply transitive $3$-dimensional groups $G_3$, referred to as the \emph{Bianchi groups}.\footnote{The classification of minimal geometries made by Kodama in \cite[Table~1]{1998_Kodama} and \cite[Table~1]{2002_Kodama} is not coherent between these two papers, and differs from \Cref{thm_minimal}. While no precise definition of ``minimal geometry'' is given in these papers, the reason for the difference is likely because the author referred to minimal geometries as ``the minimal groups admitting compact quotient with a discrete subgroup $\Gamma \subset \Gmin$.'' Here, we do not require the existence of a $\Gamma \subset \Gmin$ such that $\Sigma = \tilde\Sigma/\Gamma$. Hence, we follow the definition in \cite{1994_Koike_et_al, 1997_Tanimoto_et_al}.\label{footnote_Kodama}} They are classified from ${\rm I}$ to ${\rm IX}$ (\cite{1898_Bianchi, 2001_Bianchi_Jantzen}), and we denote them by $\BI$ to $\BIX$.
\begin{theorem}[\cite{1985_Fagundes, 1994_Koike_et_al, 1997_Tanimoto_et_al}]\label{thm_minimal}
The minimal subgeometries of the eight maximal geometries are for
\begin{enumerate}
    \item the $\mE^3$-geometry: $\Gmin = \BI$ and $\Gmin = \BVIIn$,

    \item the $\mS^3$-geometry: $\Gmin = \BIX$,

    \item the $\mH^3$-geometry: $\Gmin = \BV$ and $\Gmin = \BVIIh$,

    \item the $\RSS$-geometry: $\Gmin = \Gmax$,

    \item the $\RHH$-geometry: $\Gmin = \BIII$,
    
    \item the $\Nils$-geometry: $\Gmin = \BII$,
    
    \item the $\Sols$-geometry: $\Gmin = \BVIn$,

    \item the $\SLRRs$-geometry: $\Gmin = \BVIII$ and $\Gmin = \BIII$.
\end{enumerate}
\end{theorem}
While  $\BVIh$ and $\BIV$ groups define minimal geometries, they are not subgeometries of any maximal geometry. For this reason, no metric defined on a closed manifold can be locally invariant by these groups. Hence, for the rest of the paper, we will not consider these groups. 

\begin{remark}
For the $\mH^3$ and $\RHH$ geometries, the discrete subgroups $\Gamma \subset \Gmax$ used to construct closed manifolds can never be subsets of the minimal groups. In this sense, $\BV$, $\BVIIh$, and $\BIII$ cannot be compactified \cite{1991_Ashtekar_et_al_topology}. However, some metrics locally invariant by these minimal groups can still be defined on closed manifolds, but their isometry group will necessarily be bigger than $\Gmin$. This is discussed in \Cref{sec_Isom_Bianchi}.
\end{remark}
\begin{definition}[Minimal metric]\label[definition]{def_minimal_metric}
Given a closed  geometric $3$-manifold $\Sigma$, a Riemannian metric $\T h$ on $\Sigma$ is said to be \emph{minimal}, or \emph{Bianchi--Kantowski--Sachs} \emph{(BKS)}, if its (local) isometry group $\Isom(\tilde\Sigma,\T h)$ includes (one of) the minimal group $\Gmin$ of the maximal geometry on which $\Sigma$ is modeled. 
\end{definition}
\Cref{def_minimal_metric} implies that any maximal metric is also minimal, but the reverse is not necessarily true. In \Cref{sec_Isom_Bianchi}, we detail which minimal metrics are necessarily maximal. This is summarized in the third column of \Cref{tab_Gmax_Gmin}.

\subsection{Minimal metrics as left-invariant metrics on Bianchi groups}\label{sec_Bianchi}

\begin{table}[t!]
\centering\small
\caption{\small Classification of Bianchi groups with respect to Milnor bases. Note that we do not consider here that these bases are orthonormal with respect to a given metric. That is why the values of the structure constants (for a given sign) are totally free, i.e, not constrained by compactness considerations.\label{tab_Bianchi_groups}}
\renewcommand{\arraystretch}{2}
\renewcommand{\tabcolsep}{10pt}
\newcommand{\vspacetable}[1]{\multicolumn{2}{c}{} \vspace{#1}\\}

\begin{tabular}{llcccc}

\toprule
\makecell{Maximal \\ geometry}
& \makecell[l]{\; Bianchi \\ \; group}
& $n_1$
& $n_2$
& $n_3$
& $a$
\\
\midrule

$\mE^3$
& $\begin{dcases} \BI \\ \BVIIn \end{dcases}$
& $0$
& $\begin{dcases} 0 \\ + \end{dcases}$
& $\begin{dcases} 0 \\ + \end{dcases}$
& $0$
\\

$\mS^3$
& $\ \ \: \BIX$
& $+$
& $+$
& $+$
& $0$
\\

$\mH^3$
& $\begin{dcases} \BV \\ \BVIIh \end{dcases}$
& $0$
& $\begin{dcases} 0 \\ + \end{dcases}$
& $\begin{dcases} 0 \\ \tfrac{a^2}{h \,n_2} > 0 \end{dcases}$
& $+$
\\

$\RHH$
& $\ \ \: \BIII$
& $0$
& $+$
& $-\tfrac{a^2}{n_2}<0$
& $+$
\\

$\Nils$
& $\ \ \: \BII$
& $+$
& $0$
& $0$
& $0$
\\

$\Sols$
& $\ \ \: \BVIn$
& $0$
& $+$
& $-$
& $0$
\\

$\SLRRs$
& \makecell[l]{$\ \ \: \BVIII$ \\ \ (and $\BIII$)}
& \makecell[c]{$-$ \\ \ }
& \makecell[c]{$+$ \\ \ }
& \makecell[c]{$+$ \\ \ }
& \makecell[c]{$0$ \\ \ }
\\
\bottomrule

\end{tabular}
\end{table}

When the minimal group $\Gmin$ is a Bianchi group $G_3$, it acts simply transitively on $X$. Thus, $X$ can be identified with the group itself, and a minimal metric is a \emph{left invariant} metric on~$G_3$. A Killing vector field from the Bianchi group, denoted by $\T\xi$, is therefore a generator of the left-translation and the metric has constant components in a left-invariant basis, denoted by $\{\e_i\}_{i\in\{1,2,3\}}$. The \emph{commutation coefficients} $C^k_{ij} \T\xi_k \coloneqq [\T{\xi}_i, \T{\xi}_j]$ of a basis $\{\T\xi_i\}_{i\in\{1,2,3\}}$ uniquely determine the Bianchi group. They can always be written in the~form
\begin{align}
    C^k_{ij} = \epsilon_{ij\ell} n^{\ell k} + 2\delta^k_{[i} a_{j]}\,,
\end{align}
where $n^{ij}$ is a symmetric matrix, and such that $n^{ij} a_j = 0$ due to the Jacobi identities on~$C^k_{ij}$.\footnote{$n^{ij}$ and $a_i$ define a symmetric $(2,0)$-tensor density of weight $1$ and a $1$-form field globally on the group, respectively. However, $a_i$ can never descend to a global quantity on a compact manifold.} We can always choose a basis $\{\T\xi_i\}$ such that (cf., e.g., \cite{1969_Ellis_et_al, 1980_KramerStephani_exact, 1997_Wainwright_et_al_BOOK})
\begin{equation}\label{eq_Bianchi_n_i_a}
	(n^{ij}) = {\rm diag}(n_1,n_2,n_3)\,, \\  (a_i) = (a,0,0)\,, \\ a\,n_1 = 0\,, \\\ a \geq 0\,.
\end{equation}
The signs of $n_i$, $a$ and $h\coloneqq a^2/(n_2 n_3)$ determine the $G_3$ group, as presented in Table~\ref{tab_Bianchi_groups}.

The (reciprocal) algebra of any left-invariant basis $\{\e_i\}$ on $X$, i.e., $[\T\xi_i, \e_j] = 0 \ \forall \  i,j\in\{1,2,3\}$, is isomorphic to the algebra $\mathfrak{g}_3$ \cite{1969_Ellis_et_al}. Moreover, denoting by $\{\e^i\}$ the dual to the left-invariant basis, the Maurer--Cartan equation reads 
\begin{equation}\label{eq_Maurer_Cartan}
    \T\dd \e^i = - \frac{1}{2} C^i_{jk} \e^j \wedge \e^k\,.
\end{equation}
Therefore, we can always choose a left-invariant basis whose commutation coefficients correspond to~\eqref{eq_Bianchi_n_i_a}. Such a basis is called a \emph{Milnor basis} \cite{1976_Milnor}. We can always choose a Milnor basis that diagonalizes any left-invariant Riemannian metric $\T h$ \cite{1969_Ellis_et_al}. Two choices of such a Milnor basis will be of interest for this paper:
\begin{enumerate}
	\item the basis is \emph{orthonormal} with respect to $\T h$, i.e., $\T h(\e_i, \e_j) = \delta_{ij}$. In general, no freedom is left on the basis to choose the values of the constants $n_i$ and $a$, unless some $n_i$ are degenerate. When considered in a spacetime context, this approach is called the \emph{orthonormal approach}, which we detail in Section~\ref{sec_orthonormal_approach}.
	\item the constants $n_i$ and $a$ are either $1$, $-1$, or $0$, and the metric is diagonal. We call this Milnor basis \emph{unit}. Depending on the Bianchi group, some rescaling freedom of the basis vectors can remain, enabling us to choose some of the diagonal components of the metric.
\end{enumerate}
Since all maximal metrics are Bianchi metrics, except for the $\RSS$-geometry, these metrics can be expressed in an orthonormal Milnor basis with a specific choice of structure constants. We provide them in the second column of Table~\ref{tab_Thurston_Rbar_max}. We also provide their maximal Ricci tensor, calculated from the following formula (cf., e.g., \cite[Equation~(3.7b)]{1969_Ellis_et_al}
\begin{equation}\label{eq_Ricci_Bianchi}
    \sR_{ij} = 2 n_{ik} n^k{}_j - n^k{}_k n_{ij} - 2 \epsilon^{k\ell}{}_{(i}\, n_{j)k}\, a_\ell - \delta_{ij} \left[2 a_k a^k + n_{k\ell} n^{k\ell} - \tfrac{1}{2} \left(n^c{}_c\right)^2\right],
\end{equation}
where indices are raised and lowered by the metric $\T h$.


\subsection{Isometry groups of minimal metrics}\label{sec_Isom_Bianchi}

\subsubsection{Spatial isometry group}

While a minimal metric is defined such that it is $\Gmin$-invariant, this does not imply that its isometry group corresponds to $\Gmin$; indeed, additional discrete isometries (discrete isotropies) are always present in the form of the dihedral group of order $4$, denoted by $\rmD_2$ (see, e.g., \cite{1998_Kodama,2002_Kodama}). These isometries represent, in general, rotations of angle $\pi$ around the left-invariant basis vectors (for the Bianchi cases).\footnote{The dihedral group of order 8, denoted by $\rmD_4$, present in the maximal group of the $\Sols$-geometry, completes $\rmD_2$ by rotations of angle $\pi/2$ around the basis vectors. We direct the reader to the works of Kodama \cite{1998_Kodama, 2002_Kodama} for a detailed analysis of these discrete symmetries and their representations  of the Bianchi groups.} However, for some geometries, a minimal metric always has additional \emph{continuous} symmetries (continuous isotropies) such that the metric is actually maximal (see, e.g. \cite{2012_Papadopoulos_et_al}). This is the case for the $\mE^3$ (for a $\BI$-invariant metric), $\mH^3$, $\RHH$, $\Nils$ and $\RSS$-geometries. In this sense, a $\BI$, a $\BII$, a $\BIII$ (on a closed manifold), a $\BV$, a $\BVIIh$ (on a closed manifold), and a $\KS$ metric are all maximal metrics. For the other geometries, a general minimal metric  has only three continuous symmetries described by the (minimal) Bianchi group. This is summarized in the third column of Table~\ref{tab_Gmax_Gmin}, where $\Isom(\tilde\Sigma,\T h_{\Gmin})_{\rm general}$ corresponds to the local isometry group of a \emph{general} $\Gmin$-invariant metric $\T h$ defined on a closed manifold $\Sigma$.

Furthermore, as presented in Tables~\ref{tab_Gmax_Gmin} and~\ref{tab_Bianchi_groups}, some geometries ($\mE^3$, $\mH^3$ and $\SLRRs$) admit several minimal geometries.

For $\mE^3$, a minimal metric can be described either by a $\BI$ or $\BVIIn$ metric. In the former case it is always a maximal metric, while in the latter case it can be non-maximal. In this sense $\BI$ metrics are subcases of $\BVIIn$ metrics and they are equivalent when the latter is characterised by $n_2=n_3$ in an orthonormal Milnor basis. However, this inclusion does not hold when a spacetime metric is considered (cf.~\cite{1969_Ellis_et_al}, see also Sections~\ref{sec_Spacetime_Isom} and~\ref{sec_E3} below).

For $\mH^3$, a minimal metric can be described either by a $\BV$ or $\BVIIh$ metric, which differ in general. However, when we demand that the metric descend to a globally defined metric on a closed manifold (as we do in this paper), then both metrics are maximal metrics, and differ only by an homothety due to Mostow rigidity theorem \cite{1968_Mostow}.

For $\SLRRs$, a minimal metric can be described either by a $\BIII$ or $\BVIII$  metric. The $\BVIII$ metric can be fully general, however, for the $\BIII$ metric $n_2 \not=a$ is required in an orthonormal Milnor basis. For $n_2 = a$, the metric is instead defined on the $\RHH$-geometry. Interestingly, choosing $n_2 \not= a$ or $n_2 = a$ does not change the dimension of $\Isom(\tilde\Sigma,\T h)$, which is always $4$. However, the isometry group itself changes. In the former case it is $\GSLRR$ and in the latter case it is $\GRHH$. Bianchi III metrics are therefore always maximal $\RHH$ or $\SLRRs$-metrics. Consequently, in the $\SLRRs$ case, they are subcases of $\BVIII$ metrics (we provide in \Cref{app_SLRR_dico} their correspondence). However, contrary to the Euclidean case, this property also holds at the spacetime level (as shown in \cite{1969_Ellis_et_al}). This means that, when studying minimal metrics on the $\SLRRs$-geometry, one needs only to consider $\BVIII$ metrics.
\begin{remark}
    Depending on the fundamental group $\pi_1(\Sigma)$, additional continuous symmetries can be forced, even imposing a general minimal metric to be maximal. In this paper, we will not consider the dependence of the isometry group on the topology, and we direct the reader to the works by Kodama \cite{1998_Kodama, 2002_Kodama} for a detailed analysis of this~dependence.
\end{remark}

\subsubsection{Spacetime isometry group}\label{sec_Spacetime_Isom}

While some Bianchi metrics always have continuous isotropies, as shown in the previous section, these symmetries are not necessarily symmetries of the spacetime metric $\T g$. In other words, the spacetime metric does not necessarily inherit the isometry group of the spatial metric, but, in general, only the minimal group (cf.~\cite{1981_Szafron}). For an isometry of the spatial minimal metric to be an isometry of the spacetime metric, a necessary and sufficient condition is that the shear tensor $\T \sigma$ (see \Cref{sec_3+1} for the definition) is invariant by this isometry. Whether or not the shear tensor shares some isotropies with the metric depends on the constraint equations, the matter content and the fundamental group $\pi_1(\Sigma)$ (see, e.g., \cite{1998_Kodama, 2002_Kodama,2001_Barrow_et_al_b}). Thus, there is an important distinction to make between the fact that some Bianchi metrics always have more than three continuous symmetries, and the notion of \emph{locally rotationally symmetric} (LRS) metrics used in the literature on Bianchi spacetimes (see, e.g., \cite{1969_Ellis_et_al}); LRS is a property of the set $(\T h, \T\sigma)$ and not the spatial metric $\T h$ alone.
\begin{remark}
There always exists a subgroup of $\Sym(\tilde\Sigma,\T\sigma)$ that is isomorphic to $\Isom(\tilde\Sigma,\T h)$, but this does not necessarily imply $\Sym(\tilde\Sigma,\T h) \subseteq \Sym(\tilde\Sigma,\T \sigma)$. This is the case if the Killing vector generating the infinitesimal isotropy of the metric rotates with time. A simple example is that of $\BI$ metric of the form $h_{ij}\dd x^i \dd x^j = A_1(t)\dd x^2 + A_2(t)\dd y^2 + A_3(t)\dd z^2$. Its shear tensor is $\sigma_{ij}\dd x^i \dd x^j = \partial_t A_1\,\dd x^2 + \partial_t A_2\,\dd y^2 + \partial_t A_3\,\dd z^2$. The vector $\T\xi_4 = A_2(t)\,{y}\,\Tparx - A_1(t)\,{x}\,\Tpary$ is a (time-dependent) Killing vector of isotropy for the metric. However, only its time derivative is a shear collineation. Thus, both $\T h$ and $\T\sigma$ have a continuous isotropy, but the associated collineation is different. $\T\xi_4$ is a shear collineation if and only if $A_2(t) = A_1(t)$, in which case we say that the set $(\T h, \T \sigma)$ is locally isotropic (LRS) around the $z$-axis. A similar example can be built for the Nil ($\BII$) metric.
\end{remark}

\subsection{Ricci tensors of the maximal metrics}\label{sec_max_Ric_max_h}

\begin{table}[t!]
\centering\small
\caption{\small Maximal metrics and their maximal Ricci tensors, \emph{expressed in an orthonormal Milnor basis}, apart from the $\RSS$-geometry where we directly use coordinates. \label{tab_Thurston_Rbar_max}}

\renewcommand{\arraystretch}{2}
\renewcommand{\tabcolsep}{10pt}

\begin{tabular}{lcc}

\toprule
\makecell{\bf Maximal \\ \bf geometry}
& \bf \makecell[c]{ An orthonormal Milnor basis \\ of a maximal metric}
& \bf {$\sTRic$}
\arraybackslash\\
\midrule

$\mE^3$
& $n_1 = n_2 = n_3 = a = 0$
& $\T 0$
\\

$\mS^3$
& $n_1 = n_2 = n_3 > 0, \quad a = 0$
& $\dfrac{n_1^2}{2} \, {\rm diag}(1,1,1)$
\\

$\mH^3$
& $n_1 = n_2 = n_3 = 0, \quad a >0$
& $-2 \, a^2\, {\rm diag}(1,1,1)$
\\

$\RHH$
& $n_1 = 0,\quad n_2 = -n_3 = a > 0$
& \addstackgap{$- 2 \,a^2\left(\begin{smallmatrix}
            -2  & 0 & 0 \\
            0 & -1 & 1 \\
            0 & 1 & -1
            \end{smallmatrix}\right)$}
\\

$\Nils$
& $n_1 > 0, \quad  n_2 = n_3 = a = 0$
& $\dfrac{n_1^2}{2} \, {\rm diag}(1,-1,-1)$
\\

$\Sols$
& $n_1 = 0, \quad n_2 = - n_3 > 0, \quad a = 0$
& $-2 \, n_2^2\, {\rm diag}(1,0,0)$
\\

$\SLRRs$
& $n_1 < 0,  \quad n_2 = n_3 > 0,  \quad a = 0$
& \makecell[c]{ (Maximal iff $n_1 = -2\,n_2$) \\ $2\,n_2^2\, {\rm diag}(1,-2,-2)$}
\\

$\RSS$
& $A^2 \, \left(\dd x^2 + \sin^2 x\,  \dd y^2\right) + \dd z^2$
& $\dd x^2 + \sin^2 x\,  \dd y^2$
\\
\bottomrule

\end{tabular}
\end{table}

The Ricci tensor of maximal metrics is a building block of the topo-GR theory, as we will present in Section~\ref{sec_topo_GR} below. In this section, we summarize some properties of this tensor.

Up to isometries, a maximal geometry can sometimes define a family of maximal metrics on the universal covering space $\tilde\Sigma = X$, rather than a single metric. For maximal metrics described as Bianchi metrics, this family can be characterized by the structure constants $n_i$ and $a$ of an orthonormal Milnor basis. In Table~\ref{tab_Thurston_Rbar_max}, we provide the maximal metrics in terms of choice of $n_i$ and $a$ for each maximal geometry. For the $\mE^3$-geometry, the maximal metric is unique. For the geometries $\{\mS^3, \mH^3, \RHH, \Nils, \Sols, \RSS\}$ there is a one-parameter family of maximal metrics. For the $\SLRRs$-geometry, there is a two-parameter family of maximal metrics. However, while for the geometries other than $\SLRRs$ the parameter freedom on the Ricci tensors does not change their symmetry group, for $\SLRRs$ it does. As shown in \Cref{app_SL2R}, when $n_1\not=2\,n_2$, $\sTRic$ has four continuous symmetries (the same as those of its metric), but when $n_1 = -2\,n_2$, $\sTRic$ has six continuous symmetries (of course the metric still has only four of them). Therefore, for the $\SLRRs$-geometry, there is a subset of maximal metrics that maximizes the symmetry group of the Ricci tensor.\footnote{Interestingly, this tensor can also be obtained from a non-maximal $\BVIII$ metric, as shown in Appendix~\ref{app_SL2R}.} When constructing topo-GR, this is that maximal Ricci tensor which we will consider. We provide the Ricci tensors of the maximal metrics in the third column of \Cref{tab_Thurston_Rbar_max}.

Note that while it seems that the Ricci tensors for some geometries are not unique due to the presence of a structure constant, this is not the case. That structure constant, which always appears as a conformal factor, can be absorbed by a change of left-invariant basis, from the orthonormal basis of the maximal metric to a unit Milnor basis (i.e., in which the structure constants are $1$, $-1$ or $0$). One can do so by a transformation proportional to the identity, e.g., for $\mS^3$ by $\e_i \mapsto \e_i /n_1$. No free parameter is left in the commutation coefficients of the basis, nor in the components of the Ricci tensors.  Consequently, while a maximal geometry can define, up to isometries, a family of maximal metrics on $X$, it always defines a \emph{unique} maximal Ricci tensor. As will be detailed in \Cref{sec_topo_GR_Rbar_def}, this uniqueness is the reason why topo-GR can be seen as a parameter-free modification of~GR.

Finally, for topo-GR, the following property will also be of use: for all the Ricci tensors listed in \Cref{tab_Thurston_Rbar_max}, any vector field in their kernel is a collineation, i.e.,
\begin{equation}\label{eq_property_kernel_Rbar}
	\Lie{\T v} {\sTRic} = 0\,; \quad \forall \, \T v \in \ker\left({\sTRic}\right).
\end{equation}

\section{Topo-GR}\label{sec_topo_GR}

\subsection{Field equations }\label{sec_topo_GR_Field_Equations}

The theory developed in \cite{2024_Vigneron} features two geometric structures on a 4-manifold~$\CM$:
\begin{enumerate}
    \item a (physical) Lorentzian metric $\T g$, with its Levi-Civita connection $\stTnabla$, its Riemann tensor $\stTRiem$ and its Ricci tensor $\stTRic$.
    \item a \emph{reference}
    (or background) torsion-free connection $\stTbarnabla$, with its Riemann tensor $\stTbarRiem$ and its Ricci tensor $\stTbarRic$. This second connection is non-dynamical and is fixed solely by topological considerations (detailed below), independently of the metric $\T g$ or the matter content. Notably, $\stTbarnabla$ is \emph{a priori} not assumed to arise from any underlying additional Lorentzian metric.
\end{enumerate}
The presence of the additional connection $\stTbarnabla$ places this theory under the broad family of modified gravity theories known as {\it Metric Affine Gravity} (cf., e.g.,~\cite{1995_Hehl_et_al}). But, while in these theories the connection $\stTbarnabla$ is often considered flat, leading to {\it General Teleparallel Gravity} (cf., e.g.,~\cite{2020_Beltran_Jimenez_et_al}), a central idea in the approach of~\cite{2024_Vigneron} is that $\iR{4\,}{\TbarRiem}$ depends on the topology of $\CM$. This is detailed in \Cref{sec_topo_GR_Rbar_def} below.

The Lagrangian of topo-GR is given by 
\begin{equation}\label{eq_action}
	S	\coloneqq \int_\CM\sqrt{-g} \, \left[\frac{1}{2\kappa}\left(\stR_{\mu\nu} - {\stbarR_{\mu\nu}}\right) g^{\mu\nu} + \CL_{\rm m}\right] \dd^4 x\,. 
\end{equation}
Assuming that the matter Lagrangian does not depend on $\stTbarnabla$, and with $\delta\stbarR_{\mu\nu} = 0$, i.e., $\stTbarnabla$ is a background/non-dynamical connection, we obtain from this Lagrangian the following field equations
\begin{align}
	\stR_{\alpha\beta} - \stbarR_{\alpha\beta}
	    &= \kappa \, \left(T_{\alpha\beta} - \frac{T}{2} g_{\alpha\beta} \right) + \Lambda g_{\alpha\beta}\,, \label{eq_topo_GR_eq_1} \\
	\stnabla^\mu \stbarR_{\mu\alpha}
	    &= \frac{1}{2} \stnabla_\alpha \left( g^{\mu\nu} \stbarR_{\mu\nu}\right),\label{eq_topo_GR_eq_2}
\end{align}
where $\kappa \coloneqq 8\pi G$.

The action~\eqref{eq_action} corresponds to a modified Hilbert action that can be obtained with the procedure $\stR_{\mu\nu} \rightarrow \stR_{\mu\nu} - \stbarR_{\mu\nu}$ from the Hilbert action. The same applies for the field equation~\eqref{eq_topo_GR_eq_1}. The second field equation~\eqref{eq_topo_GR_eq_2} corresponds to the conservation of the term $\stbarR_{\mu\nu}  -  g_{\mu\nu} g^{\alpha\beta} \stbarR_{\alpha\beta} / 2$. Very importantly, the two field equations depend on the reference connection \emph{only} via its Ricci curvature $\stTbarRic$.

\subsection{Reference Ricci tensor }\label{sec_topo_GR_Rbar_def}

A fundamental aspect of topo-GR is the presence of a \emph{reference} Ricci tensor $\stTbarRic$. In this section, we first provide its definition, and then consider its consequence.

We assume the $4$-manifold $\CM$ to be of the form $\CM \cong \mR\times\Sigma$ with $\Sigma$ being a closed geometric $3$-manifold. Then, we define $\stTbarRic$ by the following two conditions:
\begin{enumerate}
    \item it is the Ricci tensor of a \emph{reference} connection $\stTbarnabla$ that is complete and torsion-free, and is characterized using the \emph{external Whitney sum} of a connection $\Tbarnabla_\mR$ on $\mR$ and a connection $\Tbarnabla_\Sigma$ on~$\Sigma$, i.e.,
    \begin{equation}\label{eq_def_nabla_bar}
    	\stTbarnabla = \Tbarnabla_\mR \boxplus \Tbarnabla_\Sigma\,,
    \end{equation}

    \item $\TRic[\Tbarnabla_\Sigma]$ is the maximal Ricci tensor arising from a maximal Riemannian metric on~$\Sigma$ (cf.~\Cref{tab_Thurston_Rbar_max}).
\end{enumerate}

The external Whitney sum of the two tangent bundles $T \mR$ and $T \Sigma$ is $T \mR \boxplus T \Sigma := {\rm pr}_1^* (T \mR) \oplus {\rm pr}_2^* (T \Sigma)$ with the Whitney sum of the connections as $\Tbarnabla_\mR \boxplus \Tbarnabla_\Sigma := {\rm pr}_1^* (\Tbarnabla_\mR) \oplus {\rm pr}_2^* (\Tbarnabla_\Sigma)$ where ${\rm pr}_1: \mR \times \Sigma \rightarrow \mR$ and ${\rm pr}_2: \mR \times \Sigma \rightarrow \Sigma$ (see, e.g., \cite[Chapter~I]{1978_Karoubi_BOOK} for more details on the external Whitney sum). Hence, $\stTbarnabla$ given in \eqref{eq_def_nabla_bar} 
acts on a given smooth vector field $\T Y = (\T Y_{\mR}, \T Y_\Sigma)$ along an arbitrary smooth vector field $\T X = (\T X_{\mR}, \T X_\Sigma)$ as  
$\stTbarnabla_{\T X} \T Y =  \Tbarnabla_{\T X_\mR} \T Y_{\mR} + \Tbarnabla_{\T X_\Sigma} \T Y_\Sigma$. As a result, we have ${\TRiem}[\stTbarnabla] = \TRiem[\Tbarnabla_\mR] \boxplus \TRiem[\Tbarnabla_\Sigma] =  {\rm pr}_1^* (\TRiem[\Tbarnabla_\mR]) \oplus {\rm pr}_2^*(\TRiem[\Tbarnabla_\Sigma])$, and in particular
\begin{equation}\label{eq_barRiem_sum}
    {\TRic}\left[\stTbarnabla\right] 
        = {\rm pr}_1^* \left(\TRic\left[\Tbarnabla_\mR\right]\right) \oplus {\rm pr}_2^*\left(\TRic\left[\Tbarnabla_\Sigma\right]\right),
\end{equation}
which acts as ${\TRic}[\stTbarnabla](\T X, \T Y) = \TRic[\Tbarnabla_\mR](\T X_{\mR}, \T Y_{\mR}) + \TRic[\Tbarnabla_\Sigma](\T X_\Sigma, \T Y_\Sigma)$. Since $\mR$ is 1-dimensional we have $\TRiem[\Tbarnabla_\mR] = \T 0$, and ${\TRiem}[\stTbarnabla]$ is given entirely by ${\TRiem}[\Tbarnabla_\Sigma]$. For this reason we will drop the subscript~$_\Sigma$ for the connection on $\Sigma$, and use the following notations for the curvature tensors $\stTbarRiem \coloneqq {\TRiem}[\stTbarnabla]$, $\stTbarRic \coloneqq {\TRic}[\stTbarnabla]$, ${\sTbarRiem} \coloneqq {\TRiem}[\Tbarnabla_\Sigma]$ and ${\sTbarRic} \coloneqq {\TRic}[\Tbarnabla_\Sigma]$.

Therefore, by $(ii)$, $\sTbarRic$ is uniquely defined by the universal covering space of~$\Sigma$. In this sense $\sTbarRic$, and therefore $\stTbarRic$, depend on the (spatial) topology. A direct consequence of this definition is that for Euclidean topologies for which $\tilde\Sigma = \mE^3$ and ${\sTbarRic} = \T 0$, hence $\stTbarRic = \T 0$,  topo-GR and GR are equivalent. Therefore, the two theories only differ in non-Euclidean topologies.

From the definition~\eqref{eq_def_nabla_bar}, there exists a (reference) $\Sigma$-foliation, an adapted vector basis $\{\T{\bar e}_0, \T{e}_i\}$ and dual 1-form basis $\{\T{\bar e}^0, \T{e}^i\}$, i.e., with $\T{\bar e}^0$ orthogonal to the foliation and $\T{e}_i$ tangent to the foliation, such that the components of the reference Riemann tensor are purely spatial, i.e.,
\begin{align}
    \stbarR^{ \alpha}{}_{ \beta  \mu  \nu} =
        \delta^{\alpha}_{ i} \delta^{ j}_{ \beta} \delta^{ k}_{ \mu} \delta^{\ell}_{ \nu}\, {\sbarR}^{ i}{}_{ j  k  \ell}\,. \label{eq_Riembar_ij}
\end{align}
In this same basis, the reference Ricci tensor has the form
\begin{align}
    \stbarR_{\mu\nu} =
        \begin{pmatrix}
            0 & 0 \\
            0 & {\sbarR}_{ij}
        \end{pmatrix}.\label{eq_Rbar_ij}
\end{align}
Furthermore, the Whitney sum \eqref{eq_barRiem_sum} implies that $\bare_0$ is a collineation for $\stTbarRiem$, i.e.,
\begin{align}
    \Lie{\T{\bar e}_0} \stTbarRiem = \T 0 \quad \ \Rightarrow  \quad \Lie{\T{\bar e}_0} \stTbarRic = \T 0\,. \label{eq_Lie_Riembar}
\end{align}

Relation~\eqref{eq_Rbar_ij} implies that $\T{\bar e_0} \in\kerbar$ and that $\kerbar = \{\T{\bar e}_0\} \oplus \ker\left({\sTbarRic}\right)$. Therefore, from properties~\eqref{eq_property_kernel_Rbar} and~\eqref{eq_Lie_Riembar} we have
\begin{equation}
        \Lie{\T{u}} \stTbarRic= \T 0\,; \quad \forall \, \T{u}\in\kerbar. \label{eq_topo_GR_eq_3}
\end{equation}
Since only the reference Ricci curvature is present in the field equations~\eqref{eq_topo_GR_eq_1} and~\eqref{eq_topo_GR_eq_2}, then only equations~\eqref{eq_Rbar_ij} and~\eqref{eq_topo_GR_eq_3} are relevant for topo-GR. They will be especially important when deriving the 3+1-equations of the theory.

\subsection{Discussions}

Before diving into the 3+1-decomposition of topo-GR, we provide in this section some interpretations and discussions of the topo-GR framework presented so far.

\subsubsection{Motivations for topo-GR and the form of the reference Ricci tensor} 

The initial motivation for topo-GR is to have a theory that admits a Newtonian limit in any topology, see \cite{2024_Vigneron}. We detail in this paper the motivations for this requirement, especially a few conditions are listed for a modified gravity theory and the additional term present in the Einstein equation to admit this limit. In particular, the new term should not depend on the matter content, but it has to depend explicitly on the spatial topology in a form similar to~\eqref{eq_Rbar_ij}. Taking this term to be a (reference) Ricci tensor of the form~\eqref{eq_Rbar_ij} is the simplest solution we found, essentially because it does not explicitly introduce new parameters. On top of admitting a Newtonian limit in any topology, it was shown in \cite{2024_Vigneron_et_al_b} that topo-GR allows for the construction of a simple curved inflation model, something not possible in GR, thus providing a second---here model-dependent---motivation for the~theory.

\subsubsection{Interpretation of the field equation} 

From the field equation~\eqref{eq_topo_GR_eq_1}, in the absence of matter, the spacetime is not Ricci-flat as in GR, but rather we have $\stTRic = \stTbarRic$ which is non-zero for non-Euclidean types of topologies. Therefore, the role of matter is not to directly curve spacetime as in GR, but only to induce a deviation from the ``topological curvature'' $\stTbarRic$. This was rendered in~\cite[Figure~1]{2024_Vigneron_et_al_b}.

\subsubsection{No reference foliation} 

By definition, $\stTbarnabla$ and $\stTbarRiem$ define a reference foliation, or a set of reference foliations, which can be reconstructed from a 1-form~$\T n$ satisfying $n_\alpha \stbarR^{ \alpha}{}_{ \beta  \mu  \nu} = 0$. However, the information on this foliation is lost when only the reference Ricci tensor $\stTbarRic$ is considered, since no preferred 1-form can be constructed from $\stTbarRic$ alone, i.e., the kernel of $\stTbarRic$ only involves vectors and not 1-forms. Consequently, the field equations of topo-GR, which only feature the reference Ricci curvature, and not the reference Riemann curvature,  do  \emph{not} involve a reference foliation. As shown in Section~\ref{sec_Rbar_fol}, the consequence is that in a 3+1-decomposition, $\stTbarRic$ always induces the same spatial reference Ricci tensor regardless of the choice of~foliation.

But while no preferred foliation is defined, the set $\kerbar$ defines  a \emph{reference} subset of the set of vector fields on $\CM$. The dimension of that subset depends on the type of spatial topology as we have $\dim\left[\kerbar\right] = 1 + \dim\left[\ker\left({\sTbarRic}\right)\right]$, from~\eqref{eq_Rbar_ij}. When $\Sigma$ is modeled on one of $\{\mS^3, \ \mH^3, \ \Nils, \ \SLRRs\}$-geometries, we have $\dim\left[\kerbar\right] = 1$ and a unique (up to a factor) reference vector field is defined on $\CM$. While for the $\mE^3$-geometry, we have $\stTbarRic = \T 0$ and no reference vector fields are defined on $\CM$, leading exactly to GR.
Consider the spacetime metric $\T g$ of the theory, one can construct the set of 1-forms dual to $\kerbar$, which is 1-dimensional for $\{\mS^3, \ \mH^3, \ \Nils, \ \SLRRs\}$-geometries. But these 1-forms do not necessarily fulfill the Frobenius condition needed for them to define a foliation. Therefore, even in the case $\kerbar$ is $1$-dimensional, in general, no preferred foliation can be constructed in topo-GR.

\subsubsection{Reference curvature for general manifolds $\CM$}\label{sec_topo_of_M}

The definition of $\sTbarRic$ assumes that $\Sigma$ is a geometric $3$-manifold. This is the case for the present paper, and for most cosmological models. However, if $\Sigma$ is not a geometric $3$-manifold (e.g., a connected sum), then no maximal metric can be defined on it. In this situation, the definition given in \Cref{sec_topo_GR_Rbar_def} for $\sTbarRic$ does not hold. This happens in particular for asymptotically flat spacetimes, for which spatial sections are infinite. In this situation, we define $\stTbarRic$ to be zero, and therefore, we define topo-GR to be equivalent to~GR. In this sense, the Schwarzschild metric is an exact vacuum solution of topo-GR.

For more general asymptotic conditions or when $\Sigma$ is closed but not modeled on a Thurston geometry, we do not yet have a definition for $\stTbarRic$. A possibility, that would include the current definition, would be to define $\stTbarRic$ as being the Ricci tensor of a (locally) maximally symmetric Riemannian metric on $\Sigma$, given a choice of boundary conditions (i.e., choice of closed spatial manifold $\Sigma$ or choice of asymptotic conditions). For closed non-geometric manifolds, this metric would not be locally homogeneous. However, we are not yet sure if this is a well-defined definition, leading to a unique Ricci tensor, as it does for closed geometric manifolds.

\subsubsection{A metric behind $\stTbarRic$} 

From the definition given in \Cref{sec_topo_GR_Rbar_def}, the reference connection and its curvature tensors can be obtained from a metric of the form
\begin{equation}
    {\T{\bar g}} = \T{\bar g}_\mR \boxplus \Tbarh_\Sigma \coloneqq {\rm pr}_1^* (\T g_\mR) \oplus {\rm pr}_2^*(\Tbarh_\Sigma)\,,
\end{equation}
where $\T{\bar g}_\mR$ and $\Tbarh_\Sigma$ are, respectively, any metric on $\mR$ and a maximal metric on $\Sigma$. As a consequence, there is a coordinate system in which
\begin{equation}
    \bar g_{\mu\nu} = 
        \begin{pmatrix}
            \pm 1 & 0 \\
            0 & \bar h_{ij}
        \end{pmatrix},
\end{equation}
where $\bar h_{ij}$ is time-independent. However, the field equations do not depend explicitly on this metric, which is why topo-GR is rather a ``bi-connection'' theory than a bi-metric theory.

\section{3+1-decomposition in topo-GR}\label{sec_topo_GR_3_1}

In this section, we derive for the first time the 3+1 equations of topo-GR. We provide their homogeneous form in Section~\ref{sec_topo_GR_3_1_homo}, which we will use to compute the homogeneous systems of equations of topo-GR in each Thurston geometry in Section~\ref{sec_all_models}. Additionally, we derive general results regarding LSH solutions in \Cref{sec_LSH_General_results}.

\subsection{Induced reference spatial Ricci tensor and reference tilt} \label{sec_Rbar_fol}

The orthogonal complement of $\kerbar$, denoted by $\kerbarortho$, is defined as
\begin{equation}
    \kerbarortho \coloneqq \left\{\T n \in T^*\CM\ | \ \T n(\T u) = 0 \,; \forall \, \T u \in \kerbar\right\}.
\end{equation}

\begin{proposition}
	For any $\Sigma$-foliation of $\CM$ and a normal $1$-form $\T n \not\in \kerbarortho$,  the spatial components of $\stTbarRic$ in a basis adapted to the foliation always induce the maximal Ricci tensor $\sTbarRic$ of a maximal metric of $\Sigma$.
\end{proposition}
\begin{proof}
The proof proceeds in three steps: we show that $(i)$ for a given $\T n$-foliation of $\CM$ the spatial components of $\stTbarRic$ in any adapted basis induce the same tensor on the leaves of the foliation, 
$(ii)$  this induced tensor does not depend on the choice of foliation, as long as it is a $\Sigma$-foliation (i.e., whose leaves are diffeomorphic to $\Sigma$, see \Cref{sec_notation}) and is never ``tangent'' to the kernel of $\stTbarRic$, 
$(iii)$ this unique induced tensor is the maximal Ricci tensor of a maximal metric on $\Sigma$, hence finishing the proof.

For $(i)$, it suffices to consider two bases $\{\e^0 = \T n, \e^i\}$ and $\{\T{\hat e}^0 = \T n, \T{\hat e}^i\}$ adapted to the same $\T n$-foliation that are related to each other by $\e^\mu = \Lambda^\mu{}_\nu \hat\e^\nu$. This directly implies $\stbarR_{i j} = \Lambda_i{}^k \Lambda_j{}^\ell \iR{4}{\hat{\bar R}}_{k\ell}$, showing that the spatial components of $\stTbarRic$ in two different adapted bases are related by a spatial diffeomorphism. Thus, given an $\T n$-foliation, the spatial components of $\stTbarRic$ in an adapted basis induce a unique spatial tensor on that foliation, which we denote by $\iR{\T n\, }{\sTbarRic}$.\footnote{The fact that the spatial components induce a unique spatial tensor in any adapted basis is true for any covariant tensor, but this does not hold for contravariant tensors \cite{2012_GG}.}

For $(ii)$, we consider two different $\Sigma$-foliations defined by two  $1$-forms $\e^0$ and $\hat\e^0$, respectively. Assuming that $\e^0$, $\hat\e^0$ $\notin \kerbarortho$ ensures that at any point $x\in\CM$ there exists a vector in the kernel of $\stTbarRic$ that is not tangent to these foliations (we say that the foliations are not tangent to the kernel). In turn, this ensures that there exist two adapted bases $\{\e^0, \e^i\}$ and $\{\T{\hat e}^0, \T{\hat e}^i\}$, again related by $\e^\mu = \Lambda^\mu{}_\nu \hat\e^\nu$, such that the spatial components of $\stTbarRic$ in these bases satisfy the same transformation behavior as before. Therefore, the components of the two induced tensors $\iR{\e^0\, }{\sTbarRic}$ and $\iR{\hat\e^0\, }{\sTbarRic}$ are related by $\iR{\T{\hat e}^0}{\sbarR}_{ i j} = \Lambda_{i}{}^k \Lambda_{j}{}^\ell\,\iR{\e^0}{\hat\sbarR}_{k\ell}$. Since the leaves of the $\e^0$-foliation and the $\hat\e^0$-foliation are diffeomorphic (as we assumed that they are $\Sigma$-foliations), this implies that $\iR{\e^0\, }{\sTbarRic}$ and $\iR{\hat\e^0\, }{\sTbarRic}$ are diffeomorphic as tensors defined on $\Sigma$.

For $(iii)$, the definition~\eqref{eq_Rbar_ij} of $\stTbarRic$ ensures that there exists a $\Sigma$-foliation not tangent to the kernel and such that the induced spatial reference curvature on that foliation is the maximal curvature arising from a maximal metric on $\Sigma$, which we denoted by $\sTbarRic$, as defined in Table~\ref{tab_Thurston_Rbar_max}. By uniqueness showed in $(ii)$, this concludes the proof.
\end{proof}

\begin{definition}\label[definition]{def:betabar}
Consider a spacetime $(\CM, \T g)$ and a spacelike $\Sigma$-foliation with a unit normal $1$-form $\T n \not\in \kerbarortho$. The reference tilt $\T\tilt$ is the \emph{unique} vector $\T{\bar\beta}$ fulfilling
\begin{enumerate}
   \item $\T n(\T{\bar\beta}) = 0$ $($i.e., $\T\tilt$ is spatial$)$,
   \item $\T n^\sharp + N^{-1} \T{\bar \beta} \in \kerbar$ $($i.e., $\T\tilt$ measures the tilt of $\T n^\sharp$ with respect to $\kerbar)$,
   \item $\T{\bar\beta}^\flat \in \kerbarortho$ $($i.e., ensures uniqueness of $\T\tilt$ when $\dim\left[\kerbar\right] > 1)$,
\end{enumerate}
where $N$ is the lapse of the foliation.
\end{definition}

To summarize this section, given any $\Sigma$-foliation with orthogonal 1-form fulfilling $\T n \not\in \kerbarortho$, the reference Ricci tensor $\stTbarRic$ induces the same reference spatial (maximal) Ricci tensor $\sTbarRic$ on $\Sigma$ regardless of the choice of foliation. And given a Lorentzian metric with respect to which $\T n$ is timelike, then $\stTbarRic$ induces a reference tilt $\T{\bar \beta}$ on that foliation in the direction of non-zero reference spatial Ricci curvature. That tilt depends on the choice of foliation and Lorentzian metric, which is not the case for the induced spatial reference Ricci tensor $\sTbarRic$.

\subsection{3+1-equations of topo-GR}
\label{sec_3+1}

We consider a timelike $\Sigma$-foliation with unit $1$-form $\T n \not\in \kerbarortho$. We denote by $N$ the lapse of the foliation, and $\T h \coloneqq \T g + \T n \otimes \T n$ the induced spatial metric with the Levi-Civita connection $\T\D$ and associated Riemann and Ricci tensors $\sTRiem$ and $\sTRic$, respectively, and the scalar curvature $\sR := \tr_{\T h} \sTRic$.
We introduce the expansion tensor $\T\Theta \coloneqq \Lie{\T n^\sharp} \T h / 2$ of the foliation, its trace $\theta \coloneqq \tr_{\T h} \T \Theta$, and its traceless part, i.e., the shear tensor $\T\sigma \coloneqq \T\Theta - \theta \T h/3$.

Moreover, we introduce the $3+1$-variables of the matter energy-momentum tensor with respect to $\T n^\sharp$ as $\rho$ for the density, $p$ for the pressure, $\T q$ for the momentum and $\T \pi$ for the anisotropic stress as follows (cf., e.g., \cite[Section~5.1.2]{2012_GG})
\begin{equation}\label{eq:fluidvar}
    \rho
        := n^\mu n^\nu T_{\mu\nu}\,, \\
    p
        := \frac{1}{3} h^{\mu\nu} T_{\mu\nu} \,, \\
    q_\alpha
        := - h^\mu{}_\alpha n^\nu T_{\mu\nu}\,, \\
    \pi_{\mu\nu}
        := h^\alpha{}_\mu h^\beta{}_\nu T_{\alpha\beta}
            - p h_{\mu\nu}\,.
\end{equation}
The 3+1-projection of the reference curvature can be written as functions of the reference tilt $\T{\bar\beta}$ and the reference spatial curvature $\sTbarRic$ as follows
\begin{align}
    n^\mu n^\nu \stbarR_{\mu\nu}
        &= \frac{1}{N^2} \bar\beta^i\bar\beta^j {\sbarR}_{ij}\,, \\
    h^\mu{}_i n^\nu \stbarR_{\mu\nu}
        &= - \frac{1}{N} \bar\beta^j {\sbarR}_{ji}\,, \\
    h^\mu{}_i h^\nu{}_j \stbarR_{\mu\nu}
        &= {\sbarR}_{ij}\,.
\end{align}
The effective fluid variables in the spirit of \eqref{eq:fluidvar} with respect to the $\T n$-foliation and induced by the presence of $\stTbarRic$ in equation~\eqref{eq_topo_GR_eq_1} are
\begin{align}
	\kappa\,\bar \rho
	    &= \frac{1}{2}\left(\sbarR + \frac{1}{N^2}\tilt^i \tilt^j\sbarR_{ij}\right), \\
	\kappa\,\bar p
	    &= \frac{1}{6}\left(-\sbarR + \frac{3}{N^2}\tilt^i \tilt^j\sbarR_{ij}\right), \\
	\kappa\,\bar q_i
	    &= \frac{1}{N} \tilt^j\sbarR_{ij}\,, \\
	\kappa\,\bar \pi_{ij}
	    &= \sbarR_{ij} - \frac{1}{3}\sbarR\, h_{ij}\,,
\end{align}
where $\sbarR \coloneqq \tr_{\T h} \sTbarRic$.

Then, choosing a coordinate system adapted to the foliation, the $3+1$-decomposition of the first field equation~\eqref{eq_topo_GR_eq_1} of topo-GR gives 
\begin{align}
	2\kappa \rho + 2\Lambda
        &= \frac{2}{3}\theta^2
            - \left(\sigma_{ij} \sigma^{ij}
            + \frac{1}{N^2}\tilt^i \tilt^j \sbarR_{ij}\right)
            + \sR - \sbarR\,, \label{eq_ADM_Hamiltonian} \\
	\kappa q_i
        &= \frac{2}{3}\D_i \theta
            - \D^j\sigma_{ij}
            - \frac{1}{N}\tilt^j \sbarR_{ij}\,, \label{eq_ADM_Momentum} \\
    \frac{1}{N}\left(\partial_t - \Lie{\T\beta}\right) \theta
        &= -\frac{1}{3}\theta^2
            - \sigma_{ij} \sigma^{ij}
            - \frac{1}{N^2} \tilt^i \tilt^j \sbarR_{ij}
            + \frac{1}{N} \Delta N
            - \frac{\kappa}{2} \left(\rho + 3p\right)
            + \Lambda\,, \label{eq_ADM_theta} \\
	\frac{1}{N}\left(\partial_t - \Lie{\T\beta}\right) \sigma^i{}_j
	    &= -\theta \sigma^i{}_j
	        - h^{ik}\left(
	            \sR_{\langle kj\rangle}
	            - \sbarR_{\langle kj \rangle}
	            - \frac{1}{N}\D_{\langle k} \D_{j \rangle} N
            \right)
            + \kappa\,\pi^i{}_j\,. \label{eq_ADM_sigma}
\end{align}
The space part of the second field equation~\eqref{eq_topo_GR_eq_2} leads to
\begin{equation}\label{eq_ADM_tilt_dot}
\begin{aligned}
    \left(\partial_t - \Lie{\T\beta}\right)\left(\tilt^j \sbarR_{ij}\right)
        &= \left[\left(\partial_t - \Lie{\T\beta}\right)\ln N - \theta N\right] \tilt^j \sbarR_{ij}
            - \frac{1}{2} \D_i\left(\tilt^j \tilt^k \sbarR_{jk}\right) \\
            &\qquad - N^2 \left[\D^j \left(\sbarR_{ij}
            - \frac{1}{2} \sbarR h_{ij}\right)
            + \sbarR_{ij}\, \D^j \ln N\right] .
\end{aligned}
\end{equation}
We also need to consider \eqref{eq_topo_GR_eq_3} which implies $\Lie{\T n^\sharp + N^{-1}\T \tilt}\stTbarRic = 0$ due to \Cref{def:betabar}. Only the spatial part of this equation gives an independent equation, which is
\begin{align}
    \left(\partial_t - \Lie{\T\beta-\T\tilt}\right)\sbarR_{ij} = 0\,. \label{eq_ADM_Rbar_dot}
\end{align}
Note that the projection of~\eqref{eq_topo_GR_eq_2} along $\T n^\sharp$ results in \eqref{eq_ADM_tilt_dot} and \eqref{eq_ADM_Rbar_dot}.

The system of equations~\eqref{eq_ADM_Hamiltonian}-\eqref{eq_ADM_Rbar_dot} along with the definition of $\T\Theta$ correspond to the 3+1-equations of topo-GR. Compared to GR, two additional variables are present: the reference spatial curvature $\sTbarRic$ and the reference tilt $\T{\tilt}$. The former is fully fixed by its definition as being the maximal Ricci tensor of a maximal metric on $\Sigma$ (defined in \Cref{sec_max_Ric_max_h} and in Table~\ref{tab_Thurston_Rbar_max}), and by the staticity (along the kernel of $\stTbarRic$) imposed by~\eqref{eq_ADM_Rbar_dot}. The dynamics of the latter is determined by~\eqref{eq_ADM_tilt_dot}.

By using~\eqref{eq_ADM_Rbar_dot}, equation~\eqref{eq_ADM_tilt_dot} can be replaced by
\begin{equation}
\begin{aligned}
    &\sbarR_{ij}\left[\left(\partial_t - \Lie{\T \beta} - \tilt^k \D_k \right)\tilt^j - \left[\left(\partial_t - \Lie{\T\beta}\right)\ln N - \theta N\right] \tilt^j +  N^2 \D^j\ln N\right] \\
        &\quad= - \left(N^2 h^{jk} - \tilt^j \tilt^k \right) \left(\D_j\sbarR_{ki} - \frac{1}{2} \D_i \sbarR_{jk}\right).
\end{aligned}
\end{equation}
We also obtain the following equation by contracting \eqref{eq_ADM_tilt_dot} with $\T\tilt$ and using the trace of~\eqref{eq_ADM_Rbar_dot}
\begin{equation}\label{eq_ADM_BBRbar_dot}
\begin{aligned}
    \frac{1}{2}\left(\partial_t - \Lie{\T\beta}\right)\left(\tilt^i \tilt^j \sbarR_{ij}\right)
        &= \left[\left(\partial_t - \Lie{\T\beta}\right)\ln N - \theta N\right] \, \tilt^i \tilt^j \sbarR_{ij} \\
        &\quad - N^2 \tilt^i\left[\D^j\left(\sbarR_{ij} - \tfrac{1}{2}\sbarR\, h_{ij}\right) + \sbarR_{ij} \, \D^j\ln N\right]. 
\end{aligned}
\end{equation}
This equation determines the evolution of the term $\sTbarRic (\T \tilt, \T \tilt) = \tilt^i \tilt^j \sbarR_{ij}$ present in the Hamilton constraint~\eqref{eq_ADM_Hamiltonian} and the Raychaudhuri equation~\eqref{eq_ADM_theta}.

\subsection{3+1-equations for homogeneous solutions } \label{sec_topo_GR_3_1_homo}

For a spatially homogeneous solution of the topo-GR field equations, the lapse is constrained to be $N=1$ and any scalar is spatially constant. Furthermore, since $\Sigma$ is closed, any divergence of a locally homogeneous spatial vector field vanishes. We define the expansion rate $H \coloneqq \theta/3$, and include $\Lambda$ in the energy-momentum tensor. From the system~\eqref{eq_ADM_Hamiltonian}-\eqref{eq_ADM_Rbar_dot}, we obtain the following system of equations
\begin{align}
    1
        &= \Omega_\sigma + \Omega_\tilt + \OmegaDR + \Omega_\rho\,,
        \label{eq_homo_Hamiltonian}\\
	\kappa q_i
        &= - \D^j\sigma_{ij} -  \tilt^j \sbarR_{ij}\,, \label{eq_homo_Momentum}\\
  \partial_t H
        &= -(1+q)H^2\,, \label{eq_homo_theta}\\
	\left(\partial_t - \Lie{\T\beta}\right) \sigma^i{}_j
	    &= -3H \sigma^i{}_j
	        - h^{ik}\left(\sR_{\langle kj\rangle} - \sbarR_{\langle kj \rangle} \right)
	        + \kappa\,\pi^i{}_j\,, \label{eq_homo_sigma} \\
    \left(\partial_t - \Lie{\T\beta-\T\tilt}\right)\sbarR_{ij}
        &= 0\,, \label{eq_homo_Rbar_dot} \\
    \left(\partial_t - \Lie{\T\beta}\right)\left(\tilt^j \sbarR_{ij}\right)
        &= - 3H \tilt^j \sbarR_{ij} - \D^j\sbarR_{ij}\,,
    \label{eq_homo_tilt_dot}
\end{align}
where
\begin{equation}
\begin{aligned}
	\Omega_\rho
	    &\coloneqq \frac{\kappa\rho}{3 H^2} \,, & 
    \OmegaDR
        &\coloneqq -\frac{\sR-\sbarR}{6 H^2} \,, &
    q
        &\coloneqq 2\left(\Omega_\sigma + \Omega_\tilt\right)
            + \tfrac{1}{2} \left(3w + 1\right)\Omega_\rho \,, \\
	\Omega_\sigma
	    &\coloneqq \frac{\sigma_{ij}\sigma^{ij}}{6H^2} \,, &
    \Omega_\tilt
        &\coloneqq \frac{\tilt^i\tilt^j \sbarR_{ij}}{6H^2} \,, &
    w &\coloneqq \frac{p}{\rho} \,. \\
\end{aligned}
\end{equation}
Moreover, equations~\eqref{eq_homo_sigma} and \eqref{eq_homo_tilt_dot} lead to
\begin{align}
    \frac{1}{2} \, \partial_t\left(\sigma_{ij}\sigma^{ij}\right)
        &= -3H \sigma_{ij}\sigma^{ij}
            - \sigma^{ij}\left(\sR_{ij}
            - \sbarR_{ij}  \right)
            + \kappa\, \sigma^{ij} \pi_{ij}\,,  \label{eq_homo_sigma2_dot}\\
    \frac{1}{2} \, \partial_t\left(\tilt^i\tilt^j\sbarR_{ij}\right)
        &= -3H \tilt^i\tilt^j\sbarR_{ij}
            - \tilt^i \D^j\sbarR_{ij}\,. \label{eq_homo_BBRbar_dot}
\end{align}
When $\div_{\T h} \sTbarRic = 0$, \eqref{eq_homo_BBRbar_dot} implies $\Omega_{\tilt} \propto 1/\sfac^6$, with $\dot \sfac/\sfac \coloneqq H$, which yields that $\Omega_{\tilt}$ has a contribution similar to the shear parameter of Bianchi type I models in GR.

\subsection{General results for LSH solutions}\label{sec_LSH_General_results}

We present in this section some results for LSH solutions of topo-GR that are independent of choosing a precise model, i.e., a precise Thurston geometry on which the spatial manifold is modeled.

\subsubsection{Isotropization of LSH solutions}

We can derive a theorem similar to Wald's theorem \cite{1983_Wald} for isotropization of LSH solutions in GR in the presence of a positive cosmological constant.
\begin{theorem}\label[theorem]{thm_isotropization}
	 Consider an LSH solution of topo-GR with
	\begin{enumerate}
	 	\item a positive cosmological constant
	 	\item matter satisfying the dominant and strong energy conditions,
		\item no reference tilt,
		\item $\sR-\sbarR \leq 0$.
	\end{enumerate}
	If $H>0$ initially at time $t_0$, then $H>0$ for all time $t \geq t_0$, and $\T\sigma \overset{t\rightarrow\infty}{\longrightarrow} \T 0$.
	Additionally, for non-tilted perfect fluids, we have $\sTRic - \sTbarRic \overset{t\rightarrow\infty}{\longrightarrow} \T 0$.
\end{theorem}
\begin{proof}
Assuming $\T\tilt = \T 0$, the Hamilton constraint~\eqref{eq_homo_Hamiltonian} and the Raychaudhuri equation~\eqref{eq_homo_theta} differ  from those of GR only by the presence of $\sbarR$. The proof is therefore very similar to that of Wald \cite{1983_Wald}. In \cite{1983_Wald}, the necessary condition for obtaining $\T\sigma \overset{t\rightarrow\infty}{\longrightarrow} \T 0$ was $H^2 \geq \cc/3$. From the Hamiltonian constraint of GR, this is ensured when $\sR \leq 0$. In topo-GR, $H^2 \geq \cc/3$ is ensured by the Hamilton constraint~\eqref{eq_homo_Hamiltonian} when $\sR - \sbarR \leq 0$. For non-tilted perfect fluids, we have $\T \pi = \T 0$, and $\sTRic - \sTbarRic \overset{t\rightarrow\infty}{\longrightarrow} \T 0$ follows directly from equation~\eqref{eq_homo_sigma} and $\T\sigma \overset{t\rightarrow\infty}{\longrightarrow} \T 0$.
\end{proof}
\begin{remark}
In GR, the fourth condition in \Cref{thm_isotropization} is replaced by $\sR \leq 0$. This is why for LSH solutions modeled on $\mS^3$ (i.e., Bianchi IX solutions) or $\mR \times \mS^2$ (i.e., Kantowski--Sachs solution), a positive cosmological constant does not ensure isotropization of the solution asymptotically. In topo-GR, isotropization is ensured if $\sR-\sbarR \leq 0$. As will be shown in Section~\ref{sec_all_models}, this will be fulfilled by $\BIX$ and Kantowski--Sachs metrics. However, surprisingly, for $\BII$ solutions, the condition is not ensured, unless the LRS condition is assumed~(\Cref{sec_Nil}).
\end{remark}

\subsubsection{Shear-free LSH solutions}\label{sec_shear_free}

In topo-GR, shear-free ($\T\sigma = \T 0$), non-tilted perfect fluids ($\T q = \T 0 = \T\pi$) solutions exist for all types of topologies. Indeed, for such solutions \eqref{eq_homo_Momentum} and \eqref{eq_homo_sigma} imply $\sTRic = \sTbarRic$ and $\quad \T\tilt = 0$, and the field equations become
\begin{equation}\label{eq_shear_free_eq}
    \partial_t H = -(1+q)H^2\,, \\ \Omega_\rho = 1 \,, \\ q = \frac{1}{2} (3w + 1) \,.
\end{equation}
The shear-free condition implies that the time dependence of the spatial metric $\T h$ in an \emph{Eulerian basis}, i.e., $\T\partial_t = \T n^\sharp$, is solely given by a scale factor. In turn, this implies that $\sTRic$ is time independent in the Eulerian basis, automatically fulfilling \eqref{eq_homo_Rbar_dot}. The spacetime metric has the form
\begin{equation}
    \T g = -\T\dd t\otimes \T\dd t + \sfac^2(t) \, h_{ij} (\T x) \T\dd x^i \otimes \T \dd x^j \, .
\end{equation}
The condition $\sTRic = \sTbarRic$ forces $\T h$ to be a maximal metric, unless $\Sigma$ is modeled on the $\SLRRs$-geometry. In that case, $\T h$ can have no continuous isotropies, and therefore is not necessarily maximal. This is detailed in \Cref{app_SL2R}. The field equations \eqref{eq_shear_free_eq} for the scale factor are equal to those of a flat-FLRW model in GR.

The existence of shear-free, non-tilted perfect fluid solutions for all models in topo-GR contrasts with GR for which this type of solutions only exists for Euclidean, spherical, or hyperbolic topologies. In these cases, the solutions are isotropic, as for topo-GR, but the Hamilton constraint in GR still features the curvature parameter. Shear-free solutions in other topologies in GR require the presence of a fluid anisotropic stress \cite{1993_Mimoso_et_al} due to the anisotropic Ricci curvature sourcing the shear. Generally,   in these GR models, the anisotropic stress is just \emph{ad-hoc}, chosen to be equal to the anisotropic part of the Ricci tensor, but without a physical  model of matter giving rise to that stress (see, e.g., \cite{2011_Koivisto_et_al} for a proposition of such a physical model). In particular, when  perturbations of the shear-free LSH metric are considered, the anisotropic stress remains a pure background quantity~\cite{2012_Pereira_et_al, 2015_Pereira_et_al, 2016_Pereira_et_al, 2017_Franco_et_al}. Additionally, to get the shear-free solution in GR, the standard approach is to assume that the ad-hoc matter is solely characterized by the anisotropic stress (equal to the anisotropic Ricci curvature), only modifying the evolution equation for the shear, but not the other $3+1$-equations. In topo-GR, the effective density associated with $\stTbarRic$ modifies the Hamilton constraint. As a result, the evolution equations for expansion of shear-free solutions in topo-GR do not feature the scalar curvature (as can be seen in \eqref{eq_shear_free_eq} above), while they do so in GR. The direct consequence is that the evolution of the scale factor of shear-free solutions in topo-GR is equivalent to that of flat-FLRW solutions of GR for any BKS type, while in GR it depends on the model, and can feature a bounce or a~recollapse.

As a subcase of these solutions, we can consider $\rho = 0 = p$, leading to (shear-free) static-vacuum solutions, for which the metric is given by
\begin{equation}
	\T g = -\T\dd t\otimes \T\dd t + h_{ij} (\T x) \T\dd x^i \otimes \T \dd x^j \, ,
\end{equation}
where, again, $\T h$ is maximal, except for the $\SLRRs$-geometry (Appendix~\ref{app_SL2R}). The existence of SH static-vacuum solutions for any topology contrasts strongly with GR, for which this kind of solutions is only possible for Euclidean topologies with the Minkowski~metric. 

One interesting consequence of the existence of these static solutions is for inflation. In \cite{2024_Vigneron_et_al_b} it was shown that it allows for a well defined Bunch--Davies vacuum for spherical and hyperbolic topologies. This is because in the large-Hubble-radius limit, the perturbation equations become wave equations on a static background (with the Mukhanov--Sasaki equation in particular), hence allowing for canonical quantization. In GR, this is not possible exactly because LSH static-vacuum solutions do not exist for spherical and hyperbolic topologies (see, e.g., \cite{2019_Handley} and references therein). Only after additional (fine-tuned) assumptions on the curvature scale, canonical quantization becomes possible in the large-Hubble-radius limit. In this sense, the existence of static-vacuum solutions in topo-GR makes the construction of inflation and the choice of initial conditions simpler and more natural than in GR. As discussed in \cite{2024_Vigneron_et_al_b}, this result can be viewed as an additional (model-dependent) motivation for the topo-GR theory.

In this context, the fact that in topo-GR static-vacuum solutions exist for all types of topologies suggests that canonical quantization of the inflaton should be feasible in the same way demonstrated in \cite{2024_Vigneron_et_al_b}. If this is true, this would again contrast with GR, for which anisotropic topologies (i.e., different from Euclidean, spherical or hyperbolic) also suffer from the same limitations as spherical and hyperbolic topologies when it comes to quantizing the inflaton \cite{2012_Pereira_et_al,2015_Pereira_et_al}.

\section{Systems of equations in all models}\label{sec_all_models}

While in Section~\ref{sec_topo_GR_3_1_homo} we provided the general form of the 3+1-equations for LSH solutions, this system can be either simplified or put in a different form (e.g., with the orthonormal approach) when a specific type of topology is considered. Thus, the aim of this section is to compute the system of equations of LSH solutions for each Thurston families of topologies. The method we will use is detailed in \Cref{sec_Program}. The systems of equations will be expressed using the orthonormal approach (cf.~\cite{1969_Ellis_et_al}) detailed in \Cref{sec_orthonormal_approach}. The systems for each geometry are provided in \Cref{sec_E3} to \Cref{sec_RS2}. Finally, in \Cref{sec_discussion_exact}, we revisit for each model the general results obtained in \Cref{sec_LSH_General_results}.

\subsection{Relation between the physical metric and the reference curvature} \label{sec_Program}

For all models, a prerequisite to be able to compute the field equations is to relate the spatial reference Ricci curvature $\sTbarRic$ with the spatial metric $\T h$. More precisely, given a basis in which $\T h$ is known (e.g., an orthonormal basis), we need to determine the components of $\sTbarRic$, using, in particular, its definition as maximal Ricci tensor. The goal of this section is to present the procedure we will follow to tackle this problem. In this section, the maximal and minimal groups $\Gmax$ and $\Gmin$ are defined as abstract groups, provided in \Cref{tab_Gmax_Gmin}, not belonging to the diffeomorphism group of~$\tilde\Sigma$.

By definition, $\sTbarRic$ is constructed from a maximal Riemannian metric, and therefore there exists a transitive subgroup $\Gbarmax \subseteq \Sym(\tilde\Sigma, \sTbarRic)$ such that $\Gbarmax \cong \Gmax$, where $\Gmax$ is one of the maximal groups. Additionally,  there exists a transitive subgroup  $\Gbarmin \subseteq \Gbarmax$ such that $\Gbarmin \cong \Gmin$, where $\Gmin$ is one of the minimal groups compatible with $\Gmax$.

\begin{remark}
\label[remark]{rem_Gbarmax_and_metric}
    There might exist other subgroups $\Gbarmax' \subseteq \Sym(\tilde\Sigma, \sTbarRic)$ and $\Gbarmin' \subseteq \Gbarmax'$ such that $\Gbarmax' \cong \Gmax$ and $\Gbarmin' \cong \Gbarmin$, but with $\Gbarmax' \not= \Gbarmax$ and $\Gbarmin' \not= \Gbarmin$. The groups $\Gbarmax'$ and $\Gbarmax$ induce two different families of maximal metrics $\{\Tbarh'\}$ and $\{\Tbarh\}$ on $\tilde\Sigma$, both yielding the same Ricci tensor $\sTbarRic$. For example, in the case of the $\SLRRs$-geometry, we have $\Sym(\tilde\Sigma,\sTbarRic)_0 \cong \SLRR\times\SLRR$ (where the subscript $_0$ refers to the connected components to the identity of the group), and there are two families of maximal metrics giving the same Ricci tensor, left and right $\BVIII$-invariant metrics, respectively. The important point of this remark is that for any subgroups $\Gbarmax$ and $\Gbarmin$ with the above properties, we can always define a maximal metric $\Tbarh$ such that $\sTbarRic = \sTRic[\Tbarh]$ and for which $\Isom(\tilde\Sigma,\Tbarh) = \Gbarmax$ and $\Gbarmin \subseteq \Isom(\tilde\Sigma,\Tbarh)$.
\end{remark}

The physical metric $\T h$ is defined to be a minimal metric. Therefore, there exists a transitive subgroup $\Ghatmin \subseteq \Isom(\tilde\Sigma, \T h)$ such that $\Ghatmin \cong \Gmin$.
Thus, by definition, both $\sTbarRic$ and $\T h$ are invariant by a transitive subgroup of diffeomorphisms (of $\tilde\Sigma$) isomorphic to a minimal group, denoted by $\Gbarmin$ and $\Ghatmin$, , respectively, and for a given choice of Thurston geometry we have $\Gbarmin \cong \Ghatmin$. However, since $\Gbarmin$ and $\Ghatmin$ are isomorphic as subsets of $\Diff(\tilde\Sigma)$, in general, they are equal only up to a diffeomorphism. A direct consequence, taking the example of Bianchi groups, is that, in general, a left invariant basis for $\Ghatmin$ is not necessarily a left invariant basis for $\Gbarmin$, and vice versa. Therefore, in a basis where the components of $\T h$ are spatially constants, the components of $\sTbarRic$ are not necessarily spatially constants, even though the two tensors are locally homogeneous.  Therefore, without additional constraints on the relation between the symmetries of $\sTbarRic$ and $\T h$, a full diffeomorphism freedom remains between the components of these two tensors in a given basis. As presented just below, those constraints come from the field equations.

We assume now that the transitive subgroup $\Ghatmin \subset \Diff(\tilde\Sigma)$ making the spatial metric locally homogeneous also makes the energy-momentum tensor locally homogeneous, i.e., $\Ghatmin \subseteq \Sym(\tilde\Sigma, \rho, p, \T q, \T\pi)$.  This is the standard assumptions for LSH models in cosmology. Then, from the properties $\Ghatmin \subseteq \Sym(\tilde\Sigma, \T \sigma)$ and $\GhatBKS \subseteq \Sym(\tilde\Sigma, \Lie{\T n^\sharp} \T \sigma)$ (see \Cref{sec_Isom_Bianchi}), the equations \eqref{eq_homo_Hamiltonian} and \eqref{eq_homo_BBRbar_dot} directly imply $\Ghatmin \subseteq \Sym(\tilde\Sigma,\sTbarRic)$, leading to
\begin{equation}\label{eq_GBKS_GBKS}
    \Ghatmin = \Gbarmin \,.
\end{equation}
Therefore, from the remark~\ref{rem_Gbarmax_and_metric}, we can consider a maximal metric $\Tbarh$ such that $\sTbarRic = \sTRic[\Tbarh]$  which is invariant by the group $\Ghatmin$. We shall show in Sections~\ref{sec_E3} to \ref{sec_RS2} that, except for the $\Nils$-geometry, this property is sufficient to fully determine the components of $\sTbarRic$ as functions of those of $\T h$. Note that \emph{a priori} $\Isom(\tilde\Sigma, \T h) \subseteq \Isom(\tilde\Sigma, \Tbarh)$ does not hold and therefore $\Isom(\tilde\Sigma, \T h) \subseteq \Sym(\tilde\Sigma, \sTbarRic)$ does not hold either. Hence, as for the shear tensor (see \Cref{sec_Isom_Bianchi}), in models where the spatial metric has more continuous symmetries than that described by $\Ghatmin$, those symmetries are not necessarily symmetries of the reference Ricci tensor. We discuss this in \Cref{sec_discussion_Nil}.

Now, as mentioned in the beginning of this section, the goal is to determine the components of $\sTbarRic$ with respect to those of $\T h$. For Bianchi metrics, this amounts to answering the following question.
\begin{question}\label[question]{q_milnor}
What is the expression of~$\sTbarRic$
in an orthonormal Milnor basis for $\T h$?
\end{question}
\noindent
To answer \Cref{q_milnor}, a three-step program will be followed:
\begin{enumerate}
	\item  We consider a left-invariant basis $\{\bar\e_i\}$, in which the components of $\sTbarRic$, denoted by $\sbarR_{\bar i\bar j}$, are known, and  a unit Milnor basis in which $\sbarR_{\bar i\bar j}$ do not feature free parameters. As explained in \Cref{sec_max_Ric_max_h}, these components can easily be obtained from the third column of Table~\ref{tab_Thurston_Rbar_max} with a rescaling of the basis vectors.
	\item We determine the general transformation matrix $\Lambda$, with components $\Lambda_i{}^{\bar i}$, between the basis $\{\bar\e_i\}$ and a second Milnor basis $\{\e_i\}$ defined to be orthonormal with respect to $\T h$, i.e., 
	\begin{equation}\label{eq:etoebartrafo}
	    \e_i = \Lambda_i{}^{\bar i} \Tbar e_{\bar i} \,.
	\end{equation}
	Due to \eqref{eq_GBKS_GBKS}, $\Lambda$ is space-independent. We will decompose this matrix into an automorphism $A$ of the Lie algebra of $\{\bar\e_i\}$, and a diagonal transformation $H$. (Note that this decomposition will not be unique in general.)
	\item Using $\Lambda$, we compute the components of $\sTbarRic$ in the orthonormal basis $\{\e_i\}$. We denote them by ${\sbarR}_{ij}$. We thus have
    \begin{equation}\label{eq_Goal}
        \sbarR_{ij} = \Lambda_i{}^{\bar i} \Lambda_j{}^{\bar j} \sbarR_{\bar i\bar j}\,.
    \end{equation}
\end{enumerate}
Once this program is fulfilled, the field equations in the orthonormal basis can be computed, using the orthonormal approach detailed in the next section.

For the $\RSS$-geometry, no $\GRSS$-invariant basis can be defined. For this geometry we will directly use a canonical coordinate system to determine the components of $\sTbarRic$ and~$\T h$, and to express the field equations.

\subsection{Orthonormal approach for spacetime Bianchi metrics}
\label{sec_orthonormal_approach}

In Sections~\ref{sec_E3} to \ref{sec_RS2}, we derive and discuss the field equations for SH solutions of topo-GR for each Thurston geometry. To express the systems of equations for the type A Bianchi models, i.e., for $\mE^3$, $\mS^3$, $\Sols$, $\Nils$, and $\SLRRs$ geometries, we will consider the orthonormal approach which we describe in this section.

In the orthonormal approach introduced in \cite{1969_Ellis_et_al}, the components of tensor fields are expressed with respect to an orthonormal spacetime basis~$\{\e_0, \e_i\}$ having the following properties:
\begin{enumerate}
	\item the basis is adapted to the foliation of homogeneity, i.e., $\T{e}_0 = \T n^\sharp$ and $\T n(\e_i) = 0$,
	\item the spatial vector basis $\{\e _i\}$ is an orthonormal Milnor basis for the spatial metric $\T h$ with the algebra~\eqref{eq_Bianchi_n_i_a},
	\item the spatial vector basis $\{\e _i\}$ is time dependent, i.e., $[\e_0, \e_i] \not=0$. The time dependence is decomposed in expansion, with the expansion tensor $\T\Theta$, and rotation, with the \emph{Fermi rotation coefficients} (i.e., rotation of the triads with respect to a Fermi-propagated frame) defined as $\Omega^k \coloneqq \epsilon^{kij} \T h\left(\e_i, [\e_0, \e_j]\right) / 2$.\footnote{While $\T\Theta$ is an intrinsic property of the foliation of homogeneity, $\T \Omega$ is a property of the orthonormal basis, and therefore is not necessarily defined globally on $\Sigma$.} This leads to the relation
\begin{align}
	[\e_0, \e_i] = \left(-\Theta^k{}_i + \epsilon^k{}_{i\ell} \Omega^\ell\right)\T{ e}_k \,.
\end{align}
\end{enumerate}
Note that in this basis the spatial indices can be raised and lowered arbitrarily as the spatial metric components are the identity.

Further, we introduce the usual Wainwright--Hsu variables $\sigma_+$ and $\sigma_-$ defined from the diagonal components of the shear as (cf.~\cite{1989_Wainwright_et_al})
\begin{equation}
	\sigma_+ \coloneqq \frac{1}{2}\left(\sigma_{22} + \sigma_{33}\right), \\
	\sigma_- \coloneqq \frac{1}{2\sqrt{3}}\left(\sigma_{22} - \sigma_{33}\right).
\end{equation}

The Jacobi identities for the basis $\{\e_0,\e_i\}$ imply a direct relation between the off-diagonal components of the shear and the Fermi rotation coefficients
\begin{equation}
	\begin{aligned}
	(n_1+n_2) \sigma_{12} = (n_2-n_1)\Omega_3 \,, \\
	(n_2+n_3) \sigma_{23} = (n_3-n_2)\Omega_1 \,, \\
	(n_3+n_1) \sigma_{31} = (n_1-n_3)\Omega_2 \,,
	\end{aligned}
\end{equation}
and provides the evolution equations for
\begin{equation}
\begin{aligned}
	\dot n_1
		&= - n_1 \left(H + 4\sigma_+ \right) , \\
	\dot n_2
		&= - n_2 \left(H - 2\sigma_+ - 2\sqrt{3}\, \sigma_-\right) , \\
	\dot n_3
		&= - n_3\left(H - 2\sigma_+ + 2\sqrt{3}\, \sigma_- \right) ,
\end{aligned}
\end{equation}
where the over dot will denote the time derivative from now on.
Therefore, the set of variables coming from the metric and the orthonormal basis are in general
\begin{equation}\label{eq_set_orthonormal_metric}
	\left\{
	    H\,,\  n_1\,, \ n_2\,, \ n_3\,, \ \sigma_+\,, \ \sigma_-\,, \
	    \Omega_1\,, \ \Omega_2\,, \ \Omega_3
    \right\}.
\end{equation}
To this set one should add the fluid variables
\begin{equation}\label{eq_set_orthonormal_fluid}
    \left\{
	    \rho\,, \ p\,, \ q_1\,,\  q_2\,, \ q_3\,, \
	    \pi_+ \coloneqq \tfrac{1}{2}\left(\pi_{22} + \pi_{33}\right), \
	    \pi_- \coloneqq \tfrac{1}{2\sqrt{3}}\left(\pi_{22} - \pi_{33}\right), \
	    \pi_{12}\,, \ \pi_{23}\,, \ \pi_{31}
	\right\},
\end{equation}
and, in topo-GR, one should add the reference Ricci curvature and reference tilt variables
\begin{equation}\label{eq_set_orthonormal_Rbar}
    \left\{
	    \sbarR_{11}\,, \ \sbarR_{22}\,, \ \sbarR_{33}\,, \ \sbarR_{12}\,, \ \sbarR_{23}\,, \
	    \sbarR_{31}\,, \ \tilt_1\,,\  \tilt_2\,, \ \tilt_3
	\right\} .
\end{equation}
Using the expression of the Einstein equation with respect to these variables (cf., e.g., \cite{1969_Ellis_et_al,1997_Wainwright_et_al_BOOK}), the momentum constraint~\eqref{eq_homo_Momentum} and the evolution equation~\eqref{eq_homo_sigma} for the shear become, respectively,
\begin{align}
    \epsilon_{ijk} n^{j\ell}\sigma_\ell{}^k
        &= \kappa q_i + \tilt^j\sbarR_{ji}\,, \label{eq_homo_Momentum_ortho}\\
    \dot \sigma_{ij}
        &= -3H\sigma_{ij}
            + 2 \epsilon_{k\ell \langle i}\Omega^k \sigma_{j\rangle}{}^\ell
            - \left(\sR_{\langle ij \rangle}
            - \sbarR_{\langle ij \rangle}\right)
            + \kappa \pi_{ij}\,.\label{eq_homo_sigma_ortho}
\end{align}
Additionally, equations~\eqref{eq_homo_Rbar_dot} and~\eqref{eq_homo_tilt_dot} present in topo-GR become, respectively,
\begin{align}
    \partial_t \left(\tilt^j \sbarR_{ji}\right)
        &= -4 H \tilt^j \sbarR_{ji}
            - \tilt^j\sbarR_{jk}\sigma^k{}_i
            + \epsilon_{ij}{}^k \Omega^j \tilt^\ell \sbarR_{\ell k}
            + \epsilon_{ijk} n^{j\ell}\sbarR_\ell{}^k\,, \label{eq_homo_tilt_dot_ortho} \\
    \dot \sbarR_{ij}
        &= -2 H \sbarR_{ij}
            - 2 \, \sbarR_{k(i} \sigma^k{}_{j)}
            + 2\,  \epsilon_{k\ell (i}\left(\Omega^k \delta^{\ell m}
            - \tilt^k n^{\ell m} \right)\sbarR_{j)m}\,. \label{eq_homo_Rbar_dot_ortho}
\end{align}

The fluid variables~\eqref{eq_set_orthonormal_fluid} depend on the equation of state, which, in general, is not constrained by the choice of Bianchi metric. However, the reference curvature variables~\eqref{eq_set_orthonormal_Rbar} highly depend on the choice of Bianchi metric through their dependence on the Thurston geometry on which $\Sigma$ is modeled. In particular, we will show in the next sections that the off-diagonal components of $\sTbarRic$ are always zero for Bianchi type A (apart from $\Nils$-geometries, i.e., Bianchi II models) and that the diagonal components are uniquely expressed as functions of $n_1$, $n_2$, and $n_3$ (see Table~\ref{tab_R_Rbar}). Therefore, for these models the only new dynamical variable coming from the reference Ricci curvature ${\stTbarRic}$ will be $\T \tilt$. We will show that for the $\mE^3$ (for Bianchi I), $\mH^3$, $\RHH$, $\RSS$ geometries, the general solution imposes $\T q = \T 0 = \T\tilt$. For the other models, we will rather assume that relation, and we let the study of models with $\T q \neq \T 0 \neq \T\tilt$ for a future work.

\begin{remark}
    The orthonormal approach is more often considered with variables normalized by the Hubble rate $H$, especially for a dynamical analysis of the equations. However, since we will discuss $H=0$ solutions (see \Cref{sec_discussion_exact}), we will  consider the non-normalized variables.
\end{remark}
In situations where the shear is transverse, i.e., $\div_{\T h} \T\sigma = \T 0$, the shear is diagonal and $\Omega^i = 0$~\cite{1969_Ellis_et_al}, unless there are degeneracies in values of $n_i$. In that situation, off-diagonal components of $\sigma_{ij}$ can still be non-zero, but they can be chosen to vanish using a spatial rotation freedom of the orthonormal Milnor basis $\{\e_i\}$. This also holds for setting $\Omega^i$ to zero \cite{1969_Ellis_et_al}. Therefore, when the shear is transverse, the set of dynamical variables coming from the metric and the orthonormal basis reduces to
\begin{equation}\label{eq_set_orthonormal_no_tilt}
	\left\{ H\,, \ n_1\,, \ n_2\,, \ n_3\,, \ \sigma_+\,, \ \sigma_- \right\}.
\end{equation}
In topo-GR, this situation holds when $\T q = \T 0 = \T\tilt$, since $\div_{\T h} \T\sigma = \T 0$ is imposed by the momentum constraint~\eqref{eq_homo_Momentum}.

Let us emphasize an important point: because $\sTbarRic$ depends on the Thurston geometry on which $\Sigma$ is modeled, and therefore depends on the Bianchi model, then, contrary to GR, the systems of equations as functions of the structure constants $n_1$, $n_2$, and $n_3$ will not be the same among different Bianchi models in topo-GR.

\begin{table}[t]
\centering\small
\caption{\small Spatial reference Ricci curvature given, for Bianchi metrics, in the orthonormal Milnor basis of the physical metric $\T h$, and, for KS metrics, in canonical coordinates. \label{tab_R_Rbar}}
\renewcommand{\arraystretch}{1.5}
\newcommand{\vspacetable}[1]{\multicolumn{2}{c}{} \vspace{#1}\\}

\begin{tabular}{lcc}

\toprule
\centering \makecell[l]{{\bf Maximal }\\{\bf geometry}}
&  $\sTbarRic$
& ${\rm tr}_{\T h}\left(\sTRic - \sTbarRic\right)$
\arraybackslash\\
\midrule

$\mE^3$
& $\T 0 = \sTRic$
& $\begin{dcases} \BI: \ \ 0 \\ \BVIIn: \ \ \leq 0\end{dcases}$ \\

$\mS^3$
&   \addstackgap{$\tfrac{1}{2}\left(\begin{smallmatrix}
                {n_2 n_3} & 0 & 0\\
                0 & {n_3 n_1} & 0\\
                0 & 0 & {n_1 n_2}
            \end{smallmatrix}\right)$}
& $\leq 0$ \\

$\mH^3$
& \makecell[cc]{\addstackgap{$\left(\begin{smallmatrix}
                -2 a^2 & 0 & 0\\
                0 & -2 a^2  & 0\\
                0 & 0 & -2 a^2
            \end{smallmatrix}\right) = \sTRic$}}
& $0$ \\

$\RHH$
& \addstackgap{$\left(\begin{smallmatrix}
            -4a^2  & 0 & 0 \\
            0 & -2a^2 & 2a^2 \\
            0 & 2a^2 & -2a^2
            \end{smallmatrix}\right) = \sTRic$}
& $0$ \\

\vspacetable{-10pt}
$\Nils$
& \makecell[l]{for $\div_{\T h} \, \sTbarRic = \T 0$ \\
        \ \ \ \addstackgap{$\tfrac{n_1^2}{2}\left(\begin{smallmatrix}
                r_2 r_3 & 0 & 0\\
                0 & -r_2 & 0\\
                0 & 0 & -r_3
            \end{smallmatrix}\right),$} \\
            with $r_2,r_3 > 0$
        }
& \makecell[c]{
$
\begin{aligned}
    \leq 0 \,; \quad r_2 = r_3
    \\
    {\rm indefinite}; \, r_2 \neq r_3
\end{aligned}
$
}
\\
\vspacetable{-10pt}

$\Sols$
&  \addstackgap{$\left(\begin{smallmatrix}
                2n_2 n_3 & 0 & 0\\
                0 & 0 & 0\\
                0 & 0 & 0
            \end{smallmatrix}\right)$}
& $\leq 0$ \\

$\SLRRs$
& \makecell[cl]{\addstackgap{
    $2\left(\begin{smallmatrix}
                {n_2 n_3} & 0 & 0\\
                0 & {n_3 n_1} & 0\\
                0 & 0 & { n_1 n_2}
    \end{smallmatrix}\right)$}}
& $\leq 0$\\

$\RSS$
& ${\rm diag}(1,\sin^2x,0) = \sTRic$
& $0$
\\
\bottomrule

\end{tabular}
\end{table}

\subsection{\texorpdfstring{$\mE^3$-geometry}{}}
\label{sec_E3}

For $\mE^3$-geometric manifolds, we have $\sTbarRic = 0$ by definition. Therefore, GR and topo-GR are equivalent for this type of topologies. Two different minimal metrics can be considered on this geometry: $\BI$ and $\BVIIn$ metrics. For the sake of generality, we recall the systems of equations for both metrics.

The $\BI$ metric does \emph{not} allow for a tilt. The orthonormal Milnor basis is characterized by $n_1 = n_2 = n_3 = a = 0$, and the field equations are
\begin{subequations}\label{eq_syst_E3_tot_I}
\begin{align}
	1
	    &= \Omega_\rho + \Omega_\sigma\,, \label{eq_syst_E3_1_I}\\
	\dot H
	    &= -(1+q)H^2\,, \label{eq_syst_E3_2_I} \\
	\dot\sigma_+
		&= -3H\sigma_+\,, \label{eq_syst_E3_3_I}\\
	\dot \sigma_-
		&= -3H\sigma_-\,,\label{eq_syst_E3_4_I}
\end{align}
\end{subequations}
where the cosmological parameters have the form
\begin{equation}\label{eq_syst_E3_Omega_I}
	H^2\Omega_\sigma = \sigma_+^2 + \sigma_-^2\,, \\ 
	H^2\OmegaDR = 0\,.
\end{equation}
The shear parameter always evolves as $H^2\Omega_\sigma \propto 1/\sfac^6$.\\

For the $\BVIIn$ metric tilt is generally allowed, and the curvature is non-zero. The orthonormal Milnor basis is characterized by $n_1 = 0 = a$ and $n_2, n_3 >0$. For non-tilted perfect fluids, the field equations are
\begin{subequations}\label{eq_syst_E3_tot_VII0}
\begin{align}
	1
	    &= \Omega_\rho + \OmegaDR + \Omega_\sigma\,, \label{eq_syst_E3_1_VII0}\\
	\dot H
	    &= -(1+q)H^2\,, \label{eq_syst_E3_2_VII0} \\
	\dot n_2
		&= - n_2 \left(H - 2\sigma_+ - 2\sqrt{3}\, \sigma_-\right), \label{eq_syst_E3_3_VII0} \\
	\dot n_3
		&= - n_3\left(H - 2\sigma_+ + 2\sqrt{3}\, \sigma_- \right), \label{eq_syst_E3_4_VII0} \\
	\dot\sigma_+
		&= -3H\sigma_+ - \frac{1}{6} \left(n_2 - n_3\right)^2, \label{eq_syst_E3_5_VII0}\\
	\dot \sigma_-
		&= -3H\sigma_- + \frac{1}{2\sqrt{3}}\left(n_3^2 - n_2^2\right),\label{eq_syst_E3_6_VII0}
\end{align}
\end{subequations}
where the cosmological parameters have the form
\begin{equation}\label{eq_syst_E3_Omega_VII0}
	H^2\Omega_\sigma = \sigma_+^2 + \sigma_-^2\,, \\ 
	H^2\OmegaDR = -\frac{1}{2}\left(n_2-n_3\right)^2 \leq 0\,.
\end{equation}
\begin{remark}
In the literature, Bianchi VII$_0$ models are often said to be ``more general'' than Bianchi I models (e.g., \cite{1985_Barrow_et_al, 2006_Jaffe_et_al_BVIIhWMAP, 2006_Jaffe_et_al_BVIIhruledout, 2007_Pontzen_BianchiCMB, 2013_PlanckData_topology, 2016_SaadehPontzen_isotropic}). This interpretation of generality is in the sense that Bianchi VII$_0$ models feature more parameters than Bianchi I models (c.f., e.g., \cite{2006_Hervik_et_al}). However, this does not mean that the former includes the latter. Indeed, for a (spatial) $\BVIIn$ metric to be a (spatial) $\BI$ metric, one needs $n_2 = n_3$. However, the consequence for the spacetime metric is that $\sigma_- = 0$ (as can be seen from \eqref{eq_syst_E3_tot_VII0} above), enforcing the LRS. The LRS for $\BI$ metrics is not a generic feature, but a subset of the set of solutions. Therefore, when considering a minimal metric on the $\mE^3$-geometry for constructing a cosmological model, two complementary and inequivalent choices are~possible.
\end{remark}

\subsection{\texorpdfstring{$\mS^3$-geometry}{}}
\label{sec_S3}

\subsubsection{Reference curvature in the orthonormal basis}

As shown in Section~\ref{sec_max_Ric_max_h}, for $\mS^3$-geometric manifolds, there exists a unit Milnor basis $\{\bare_i\}$ with algebra
\begin{equation}\label{eq_algebra_bar_S3}
	(\bar n^{ij})  = {\rm diag}(1,1,1) \,, \\ \bar a_i = 0 \,,
\end{equation}
such that the maximal Ricci tensor is
\begin{equation}\label{eq_Rbar_ij_bar_S3}
	\sTbarRic =
	    \frac{1}{2} \left(\bare^1 \otimes \bare^1 +  \bare^2 \otimes \bare^2 + \bare^3 \otimes \bare^3\right).
\end{equation}
The (minimal) physical metric $\T h$ is a $\BIX$ metric and its {orthonormal} Milnor basis $\{\e_i\}$ is characterized by 
\begin{equation}\label{eq_algebra_S3}
	(n^{ij})  = {\rm diag}(n_1,n_2,n_3) \,, \\ a_i=0 \,,
\end{equation}
where $n_i > 0$ for $i \in \{1,2,3\}$.

The (constant) transformation matrix $\Lambda$, with components defined in \eqref{eq:etoebartrafo}, is given by
\begin{align}\label{eq_Lambda_S3}
    \Lambda
    =
    \begin{pmatrix}
       \sqrt{n_2 n_3} & 0 & 0 \\
        0 &  \sqrt{n_3 n_1} & 0 \\
        0 & 0 & \sqrt{n_1 n_2}
    \end{pmatrix} A \,,
\end{align}
where $A \in \rmO(3)$ is an automorphism of the Lie algebra~\eqref{eq_algebra_bar_S3}. This implies that the components of the reference Ricci tensor in the orthonormal basis $\{\e_i\}$ are
\begin{equation}\label{eq_Rbarij_S3}
    \sTbarRic =
        \frac{1}{2}\left(n_2 n_3 \, \e^1 \otimes \e^1 + n_3 n_1 \, \e^2 \otimes \e^2 + n_1 n_2 \, \e^3 \otimes \e^3\right).
\end{equation}
The reason why $\sTbarRic$ is solely expressed as a function of $n_i$ in the basis $\{\e_i\}$ is because its components \eqref{eq_Rbar_ij_bar_S3} in the basis $\{\bare_i\}$ are invariant under the automorphisms of the Lie algebra \eqref{eq_algebra_bar_S3}, corresponding to $3$-dimensional rotations.

Using relation \eqref{eq_Ricci_Bianchi} to express $\sTRic$ in the orthonormal basis and equation~\eqref{eq_Rbarij_S3}, we~get
\begin{align}\label{eq_Rbar_ij_S3_diff}
\begin{aligned}
    \sTRic - \sTbarRic
    	&= \frac{1}{2}\left(n_1^2 - n_2^2 - n_3^2 + n_2 n_3\right) \e^1 \otimes \e^1 \\
    	&\quad+	\frac{1}{2}\left(n_2^2 - n_3^2 - n_1^2 + n_3 n_1\right) \e^2 \otimes \e^2 \\
    	&\quad+	\frac{1}{2}\left(n_3^2 - n_1^2 - n_2^2 + n_1 n_2\right) \e^3 \otimes \e^3 \\
	&= \T 0\,; \quad {\rm iff} \ n_1 = n_2 = n_3 \,.
\end{aligned}
\end{align}
Consequently, we have
\begin{equation}\label{eq_DeltaRs_S3}
	\begin{split}
	\sR - \sbarR
		&= -\frac{1}{2} \left(n_1^2 + n_2^2 + n_3^2 - n_1 n_2 - n_1 n_3 - n_2 n_3\right) \\
		&\leq 0\,; \quad \, \forall \, n_i > 0\,, \\
		&=0\,; \quad {\rm iff} \ n_1 = n_2 = n_3 \,.
	\end{split}
\end{equation}
Therefore, in general $\sTRic \not= \sTbarRic$, and the equality holds \emph{if and only if} the Bianchi metric~$\T h$ is isotropic, hence maximal.

\subsubsection{Field equations }

For Bianchi IX models $\div_{\T h} \T\sigma \not= \T 0$, in general. This implies, through the momentum constraint~\eqref{eq_homo_Momentum}, that a general LSH solution includes a fluid tilt $\T q$, as well as a reference tilt $\T\tilt$ for topo-GR. As mentioned before, we assume $\T q = \T 0 = \T\tilt$ in the present work. Then, the momentum constraint~\eqref{eq_homo_Momentum} imposes $\div_{\T h} \T\sigma = \T 0$.

Using the orthonormal approach variables~\eqref{eq_set_orthonormal_no_tilt} and the expression of $\sTbarRic$ in the orthonormal basis~\eqref{eq_Rbarij_S3}, the field equations of LSH solutions of topo-GR for non-tilted perfect fluid and non-tilted reference Ricci curvature, on a $\mS^3$-geometric manifold read
\begin{subequations}\label{eq_syst_S3_2}
\begin{align}\label{eq_syst_S3_1}
	1
	&= \Omega_\sigma + \OmegaDR + \Omega_\rho\,,\\
  	\dot{H}
		&= -(1+q)H^2\,, \\
	\dot n_1
		&= - n_1 \left(H + 4\sigma_+ \right), \\
	\dot n_2
		&= - n_2 \left(H - 2\sigma_+ - 2\sqrt{3}\, \sigma_-\right), \\
	\dot n_3
		&= - n_3\left(H - 2\sigma_+ + 2\sqrt{3}\, \sigma_- \right), \\
	\dot\sigma_+
		&= -3H\sigma_+ + \frac{1}{12} \left(4 n_1^2 - 2n_2^2 - 2n_3^2 + 2n_2 n_3 - n_1 n_2 - n_1n_3\right),\\
	\dot \sigma_-
		&= -3H\sigma_- - \frac{1}{4\sqrt{3}}(n_3-n_2)(n_1 - 2 n_2 - 2 n_3)\,,
\end{align}
\end{subequations}
where the cosmological parameters have the form
\begin{align}\label{eq_syst_S3_Omega}
    \begin{split}
	H^2\Omega_\sigma
		&= \sigma_+^2 + \sigma_-^2\,, \\
	H^2\OmegaDR
		&= \frac{1}{12} \left(n_1^2 + n_2^2 + n_3^2 - n_1 n_2 - n_1 n_3 - n_2 n_3\right) \geq 0\,.
	\end{split}
\end{align}
As for Bianchi IX models in GR, if the metric is shear-free ($\sigma_- = \sigma_+ = 0$) and isotropic ($n_1 = n_2 = n_3$) at a given time, in topo-GR it remains so for all times too.

\subsection{\texorpdfstring{$\mH^3$-geometry}{}}
\label{sec_H3}

Due to Mostow rigidity theorem \cite{1968_Mostow}, on a closed manifold modeled on the $\mH^3$-geometry, all locally homogeneous metrics are homothetic to a maximal metric, i.e., they differ by a constant conformal factor, up to diffeomorphisms.\footnote{The maximal metric is both a $\BV$ and a $\BVIIh$ invariant metric as $\BV\subset\GH$ and $\BVIIn \subset\GH$. For $\BVIIh$, the orthonormal basis, characterized in general by $n_1 = 0$, $n_2>0$, $a>0$, $n_3 = a^2/(h\, n_2) > 0$, is restricted by $n_2 = a/\sqrt{h}$ due to compactness.} Additionally, if there exists a transitive subgroup $\hat G_3 \subset \Isom(\tilde\Sigma,\T h)$ such that $\hat G_3 \subset \Isom(\tilde\Sigma,\Tbarh)$ ($\hat G_3$ is either the $\BV$ or the $\BVIIh$ group), then the two metrics are exactly homothetic (i.e., the diffeomorphism freedom disappears). This is the situation between the reference metric defining $\sTbarRic$ and the physical metric $\T h$. Therefore, $\Tbarh = \alpha \T h$ with $\alpha > 0$. Consequently, we have $\sTbarRic = \sTRic$. Moreover, as shown in Appendix~\ref{app_H3}, any locally $\BV$ or $\BVIIh$-invariant vector field is zero and any locally $\BV$ or $\BVIIh$-invariant symmetric $(0,2)$-tensor is homothetic to the metric. Therefore, we always have $\T\tilt = \T 0$, $\T\sigma = \T 0$, $\T q = \T 0$, and $\T \pi = \T 0$, implying $\Omega_\sigma = 0 = \OmegaDR$, and the general field equations are
\begin{equation}\label{eq_syst_H3_tot}
    1 = \Omega_\sigma + \Omega_\rho\,, \\
    \dot{H} = -(1+q)H^2\,.
\end{equation}
These are the general field equations for an LSH solution on a closed hyperbolic topology in topo-GR. As for GR, no tilt is allowed and they correspond to the hyperbolic homogeneous and isotropic solution of the theory, derived in \cite{2023_Vigneron_et_al_b}. The only difference to GR is that the scalar curvature is not anymore present in the Friedmann equation.

\subsection{\texorpdfstring{$\RHH$-geometry}{}}
\label{sec_RH2}

\subsubsection{Reference curvature in the orthonormal basis}

As shown in Section~\ref{sec_max_Ric_max_h}, for $\RHH$-geometric manifolds, there exists a unit Milnor basis $\{\bare_i\}$ with algebra
 \begin{equation}\label{eq_algebra_bar_RH2}
 	(\bar n^{ij})  = {\rm diag}(0,1,-1)\,, \\ \bar a_i = (1,0,0) \,,
\end{equation}
such that the maximal Ricci tensor is
\begin{equation}
	\sTbarRic = -4\, \bare^1 \otimes \bare^1 - 2 \left(\bare^2-\bare^3\right)\otimes \left(\bare^2-\bare^3\right).
\end{equation}

The (minimal) physical metric $\T h$ is a $\BIII$ metric and its orthonormal Milnor basis $\{\e_i\}$ is characterized by
\begin{equation}\label{eq_algebra_RH2}
	(n^{ij})  = {\rm diag}(0,a,-a)\,, \\ (a_i)=(a,0,0)\,,
\end{equation}
where $a > 0$. As detailed in Section~\ref{sec_Isom_Bianchi}, compactness imposes $n_2 = a$ for $\BIII$ metrics.

The (constant) transformation matrix $\Lambda$, with components defined in \eqref{eq:etoebartrafo}, is given by
\begin{align}\label{eq_Lambda_RH2}
    \Lambda =
        \begin{pmatrix}
            a & c_1 & c_2 \\
            0 & c_4 & c_3 \\
            0 & c_3 & c_4
        \end{pmatrix},
\end{align}
where $c_i \in \mR$ for $i=1,2,3,4$ with $(c_4)^2 \neq (c_3)^2$. This implies that the components of the Ricci tensor~$\sTbarRic$ in the orthonormal basis $\{\e_i\}$ are
\begin{equation}\label{eq_Rbarij_RH2_pre}
    (\sbarR_{ij}) 
    =
    \begin{pmatrix}
        -4a^2 - 2r^2 & -2 r s & 2 r s \\
        -2 r s & -2s^2 & 2s^2 \\
        2 r s  & 2s^2 & -2s^2
    \end{pmatrix},
\end{equation}
where $r := c_1-c_2$ and $s := c_4 - c_3 \not=0$.

As shown in Appendix~\ref{app_RH2}, a locally homogeneous symmetric $(0,2)$-tensor on $\RHH$-geometric manifolds is restricted by the form \eqref{eq_TTensor_RH2_Compact} in a basis with algebra~\eqref{eq_algebra_RH2}. This restriction imposes $r = 0$ and $s^2 = a^2$, leading to
\begin{equation}\label{eq_Rbarij_RH2}
    \sTbarRic =
        -a^2\left[4\, \bare^1 \otimes \bare^1 + 2 \left(\bare^2-\bare^3\right)\otimes \left(\bare^2-\bare^3\right)\right].
\end{equation}
Consequently, we have $\sTbarRic = \sTRic$.

\subsubsection{Field equations}
\label{sec_RH2_syst}

As shown in \Cref{app_RH2}, any locally homogeneous vector field is in the kernel of the Ricci tensor. Therefore, we always have $\sTbarRic\left(\T\tilt,\cdot\right) = \T 0$. Additionally, any locally homogeneous symmetric $(0,2)$-tensor field is transverse, implying $\div_{\T h} \, \T\sigma = \T 0$. Along with the momentum constraint~\eqref{eq_homo_Momentum}, this imposes $\T q = \T 0$. Therefore, as for GR, LSH solutions on $\RHH$-geometric manifolds do not allow for tilted fluids. No reference tilt is allowed either. Furthermore, the shear is restricted to the form~\eqref{eq_TTensor_RH2_Compact} in the orthonormal basis $\{\e_i\}$ and therefore depends only on one variable. The consequence is that the physical metric can be written in the simple form (resembling that of Kantowski--Sachs by replacing $\sin^2 x$ by $\sinh^2 x$, representing the hyperbolic fibers)
\begin{equation}\label{eq_metric_Rbar_RH2}
	\T h = A^2(t)\left(\T\dd x\otimes \T\dd x + \sinh^2 x \, \T\dd y \otimes \T\dd y\right) + B^2(t)\,\T\dd z \otimes \T\dd z\,,
\end{equation}
implying that the (reference) Ricci curvature have the form
\begin{equation}
	\sTbarRic = \sTRic = -\left(\T\dd x\otimes \T\dd x + \sinh^2 x \, \T\dd y \otimes \T\dd y\right).
\end{equation}

Then, the field equations of a general LSH solution of topo-GR with perfect fluid on a geometric manifold modeled on $\RHH$ are
\begin{subequations}\label{eq_syst_RH2_tot}
\begin{align}
	1
	&= \Omega_\sigma + \Omega_\rho \,, \label{eq_syst_RH2_1}\\
  	\dot{H}
		&= -(1+q)H^2 \,, \label{eq_syst_RH2_2}\\
	\dot{\sigma}^i{}_j
		&= -3H \sigma^i{}_j \,, \label{eq_syst_RH2_3}
\end{align}
\end{subequations}
with
\begin{equation}
	H
	    = 4 \frac{\dot A}{A} + 2 \frac{\dot B}{B}\,, \quad
	\left(\sigma^i{}_j\right)
	    = \frac{2}{3}\left(\frac{\dot A}{A}- \frac{\dot B}{B}\right){\rm diag}\left(1,1,-2\right), \quad
    H^2\Omega_\sigma
        = \frac{4}{9}\left(\frac{\dot A}{A}- \frac{\dot B}{B}\right)^2.
\end{equation}
In particular, we have $\OmegaDR = 0$. These equations are similar to those in GR, with the difference that the spatial curvature is not present anymore. As a consequence, the shear parameter always evolves as $H^2\Omega_\sigma \propto 1/\sfac^6$, similar to $\BI$ metrics.

\subsection{\texorpdfstring{$\Nils$-geometry}{}}
\label{sec_Nil}

\subsubsection{Reference curvature in the orthonormal basis}
\label{sec_Nil_Rbar_proof}

As shown in Section~\ref{sec_max_Ric_max_h}, for $\Nils$-geometric manifolds, there exists a unit Milnor basis $\{\bare_i\}$ with algebra
\begin{equation}\label{eq_algebra_bar_Nil}
	(\bar n^{ij})  = {\rm diag}(1,0,0)\,, \\ \bar a_i = 0\,,
\end{equation}
such that the maximal Ricci tensor is
\begin{equation}\label{eq_Rbar_ij_bar_Nil}
	\sTbarRic = \frac{1}{2}\left(\bare^1 \otimes \bare^1 - \bare^2 \otimes \bare^2 - \bare^3 \otimes \bare^3\right).
\end{equation}
This tensor is a Lorentzian metric, and we can compute its Ricci tensor, leading to $\sTRic[\sTRic[\Tbarh]] \propto \Tbarh$. This implies that $\Isom(\tilde\Sigma,\Tbarh) = \Aff(\tilde\Sigma,\Tbarh) = \ColRic(\tilde\Sigma,\Tbarh) \cong \GNil$.\\

The (minimal) physical metric $\T h$ is a $\BII$ metric and its orthonormal Milnor basis $\{\e_i\}$ is characterized by
\begin{equation}\label{eq_algebra_Nil}
	(n^{ij})  = {\rm diag}(n_1,0,0)\,, \\ a_i = 0\,,
\end{equation}
with $n_1>0$. The basis $\{\e_i\}$ is unique up to a rotation around $\e_1$.

The transformation matrix $\Lambda$ between $\{\bare_i\}$ and $\{\e_i\}$ can be put in the form
\begin{equation}\label{eq_Lambda_Nil}
    \Lambda =
    	 \begin{pmatrix}
            n_1  & 0 & 0 \\
            0 & n_1 & 0 \\
            0 & 0 & n_1
        \end{pmatrix}
        \begin{pmatrix}
            a_2 a_3 - c_3 d_3  & 0 & 0 \\
            d_1 & a_2 & c_3 \\
            d_2 & d_3 & a_3
        \end{pmatrix},
\end{equation}
where $a_2$, $a_3$, $c_3$, $d_3$, $d_1$, $d_2$ $\in\mR$, with $a_2 a_3 - c_3 d_3 \not= 0$. The second matrix in~\eqref{eq_Lambda_Nil} is an automorphism of the Lie algebra~\eqref{eq_algebra_bar_Nil}. Unfortunately, compared to the other Bianchi type A models, the components of the reference Ricci tensor~\eqref{eq_Rbar_ij_bar_Nil} are not invariant under the automorphisms of the unit Milnor basis. The consequence is that the components of $\sTbarRic$ in the orthonormal basis $\{\e_i\}$ depend on 6 free-parameters ($a_2$, $a_3$, $c_3$, $d_3$, $d_1$ and $d_2$), and not solely on the commutation coefficient $n_1$ of the orthonormal basis. This is a very odd result, as it is different from all the other models. We discuss this in more details in Section~\ref{sec_discussion_Nil} below.

To reduce the degrees of freedom, we can use at this stage the field equation~\eqref{eq_ADM_tilt_dot} along with $\T\tilt = \T 0$ (which we impose), resulting in $\div_{\T h} \sTbarRic = \T 0$. This imposes $d_1 = 0 = d_2$. Then using the rotation freedom around $\e_1$ in the definition of the orthonormal basis, we can further set $\bar R_{23} = 0 = \bar R_{32}$ by taking $d_3 = -a_3 c_3/a_2$, which leads to
\begin{equation}\label{eq_Rbar_ij_Nil}
	\sTbarRic = \frac{n_1^2}{2}\left(r_2\, r_3\, \bare^1 \otimes \bare^1 - r_2 \, \bare^2 \otimes \bare^2 - r_3 \, \bare^3 \otimes \bare^3\right),
\end{equation}
where $r_2 := a_2^2 + c_3^2 > 0$ and $r_3 := \left(a_3/a_2\right)^2(a_2^2+c_3^2)>0$ are free parameters. Then we get
\begin{equation}\label{eq_Rbarij_Nil_diff}
    \begin{split}
    \sTRic - \sTbarRic
    	&= \frac{n_1^2}{2}\left[\left(1-r_2\, r_3\right) \bare^1 \otimes \bare^1 - \left(1-r_2\right) \bare^2 \otimes \bare^2 - \left(1-r_3\right) \bare^3 \otimes \bare^3\right],
	\end{split}
\end{equation}
and
\begin{equation}\label{eq_DeltaRs_Nil}
	\begin{split}
	{\rm tr}_{\T h}\left(\sTRic - \sTbarRic\right)
		&= -\frac{n_1^2}{2} \left(r_2-1\right)\left(r_3-1\right).
	\end{split}
\end{equation}
\begin{remark}\label[remark]{rem_II}
Since $\T h$ is maximal, we have $\Isom(\tilde\Sigma,\T h) \cong \GNil$. However, for $r_2 \not= r_3$, we only have $\Isom(\tilde\Sigma,\T h) \cong \Sym(\tilde\Sigma, \sTbarRic)$, and $\Isom(\tilde\Sigma,\T h) = \Sym(\tilde\Sigma, \sTbarRic)$ holds if and only if $r_2 = r_3$ (see \Cref{prop_property_Nil} in \Cref{app_Nil}). Additionally, $\sTRic = \sTbarRic$ if and only if $r_2=r_3=1$. Therefore, the $\Nils$-geometry is the only geometry in which, when $\T h$ is maximal, $\Isom(\tilde\Sigma,\T h) \subseteq \Sym(\tilde\Sigma, \sTbarRic)$ and $\sTRic = \sTbarRic$ are not guaranteed. We discuss this in Section~\ref{sec_discussion_Nil}.
\end{remark}

\subsubsection{Field equations: general case}

Using the momentum constraint~\eqref{eq_ADM_Momentum} with the assumption $\T q = \T 0 = \T\tilt$, we get $\Omega^2 = 0 = \Omega^3$. However, $\Omega^1$ is not zero in general. And since we already used the freedom of rotation around $\e_1$ in the orthonormal basis to set $\bar R_{23} = 0$, there is no more freedom left to set $\Omega^1 = 0$ nor $\sigma_{23}=0$. The field equation~\eqref{eq_homo_Rbar_dot}, written in the orthonormal basis, leads to 
\begin{equation}\label{eq_sigma23Omega1}
    \sigma_{23} =  \frac{r_2-r_3}{r_2+r_3} \Omega^1 \,.
\end{equation}
Therefore, the set of dynamical variables are
\begin{equation}
    \left\{ H, \ n_1, \ \sigma_+, \ \sigma_-, \ \sigma_{23}, \ r_2, \ r_3, \ \rho, \ p \right\} .
\end{equation}
The field equations are
\begin{subequations}\label{eq_syst_Nil_tot}
\begin{align}
	1
		&= \Omega_\sigma + \OmegaDR + \Omega_\rho \,,\label{eq_syst_Nil_1}\\
  	\dot H
		&= -(1+q)H^2\,, \label{eq_syst_Nil_2}\\
	\dot n_1
		&= - n_1 \left(H + 4\sigma_+ \right), \label{eq_syst_Nil_3}\\
	\dot\sigma_+
		&= -3H\sigma_+ - \frac{n_1^2}{12} \left(2 \, r_2 r_3 + r_2 + r_3 - 4\right), \label{eq_syst_Nil_4}\\
	(r_2 - r_3) \dot \sigma_-
	    &= -3H (r_2 - r_3) \sigma_-
	        - \frac{n_1^2}{4\sqrt{3}}(r_2 - r_3)^2
	        + \frac{2}{\sqrt{3}}(r_2 + r_3) \sigma_{23}^2\,, \label{eq_syst_Nil_5}\\
    (r_2 - r_3) \dot \sigma_{23} 
        &= -3H (r_2 - r_3) \sigma_{23}
            + 2\sqrt{3}\, (r_2 + r_3) \sigma_{23} \, \sigma_-\,,\label{eq_syst_Nil_8}\\
    \dot r_2
        &= 2\,r_2\left(3\sigma_+
            - \sqrt{3}\sigma_-\right), \label{eq_syst_Nil_r2}\\
    \dot r_3
        &= 2\,r_3\left(3\sigma_+
            + \sqrt{3}\sigma_-\right), \label{eq_syst_Nil_r3}
\end{align}
\end{subequations}
where by definition $r_2$, $r_3>0$, and where the cosmological parameters have the form
\begin{equation}\label{eq_syst_Nil_Omega}
	H^2\Omega_\sigma
		= \sigma_+^2 + \sigma_-^2 + \frac{\sigma_{23}^2}{9}\,, \\
	H^2\OmegaDR
		= \frac{n_1^2}{12} (r_2-1)(r_3-1)\,.
\end{equation}

Equations~\eqref{eq_syst_Nil_r2} and~\eqref{eq_syst_Nil_r3} arise from the evolution equation~\eqref{eq_homo_Rbar_dot_ortho} for the reference curvature. The $\Nils$-geometry is the only model in which that equation is not redundant with the evolution equations for the structure constant $n_i$, coming from the Jacobi identities. The reason for this is that the expression~\eqref{eq_Rbar_ij_Nil} of $\sTbarRic$ in the orthonormal basis features additional parameters in addition to $n_1$.

\subsubsection{Field equations: LRS case}\label{sec_Nil_LRS}

A Bianchi II spacetime is LRS when $\Isom(\tilde\Sigma,\T h) \subseteq \Sym(\tilde\Sigma,\T\sigma)$, i.e., the fourth Killing vector of the metric is a shear collineation.  As shown in \Cref{app_Nil} with \Cref{prop_property_Nil}, this happens if and only if there exists an orthonormal basis with Lie algebra~\eqref{eq_algebra_Nil} for which $\sigma_- = \sigma_{12} = \sigma_{13} = \sigma_{23} = 0$.\footnote{For this reason, LRS Bianchi II spacetimes are necessarily non-tilted (see, e.g., \cite{1998_Kodama, 2002_Kodama}).} In that case, the field equation~\eqref{eq_syst_Nil_5} and the relation \eqref{eq_sigma23Omega1} imply $r_2 = r_3$.

As shown in \Cref{rem_II}, if we assume $\Isom(\tilde\Sigma,\T h) \subseteq \Sym(\tilde\Sigma,\sTbarRic)$, we get $r_2 = r_3$. Then, from \eqref{eq_sigma23Omega1} it follows that $\sigma_{23} = 0$, and the field equations~\eqref{eq_syst_Nil_r2} and~\eqref{eq_syst_Nil_r3} impose $\sigma_-=0$. Therefore, LRS is equivalent as assuming $\Isom(\tilde\Sigma,\T h) \subseteq \Sym(\tilde\Sigma,\sTbarRic)$.

For an LRS Bianchi II spacetime, the system of equations is
\begin{subequations}\label{eq_syst_Nil_LRS_tot}
\begin{align}
	1
		&= \Omega_\sigma + \OmegaDR + \Omega_\rho\,,\label{eq_syst_Nil_LRS_1}\\
  	\dot H
		&= -(1+q)H^2\,, \label{eq_syst_Nil_LRS_2}\\
	\dot n_1
		&= - n_1 \left(H + 4\sigma_+ \right), \label{eq_syst_Nil_LRS_3}\\
	\dot\sigma_+
		&= -3H\sigma_+ - \frac{n_1^2}{6} (r_2+2)(r_2-1)\,, \label{eq_syst_Nil_LRS_4}\\
    \dot r_2
        &= 6\,r_2\,\sigma_+\,, \label{eq_syst_Nil_LRS_5}
\end{align}
\end{subequations}
where, by definition, $r_2>0$, and  the cosmological parameters have the form
\begin{equation}\label{eq_syst_Nil_Omega_LRS}
	H^2\Omega_\sigma
		= \sigma_+^2\,, \\
	H^2\OmegaDR
		= \frac{n_1^2}{12} (r_2-1)^2 \geq 0\,.
\end{equation}

\begin{remark}
    As shown in \cite{1998_Kodama,2002_Kodama}, some $\Nils$-topologies can impose the LRS condition. This is the case for the topological spaces named $T^3(n)/\mathrm{Z}_{k=3,4,5}$ in these papers.
\end{remark}

\subsection{\texorpdfstring{$\Sols$-geometry}{}}
\label{sec_Sol}

\subsubsection{Reference curvature in the orthonormal basis}

As shown in Section~\ref{sec_max_Ric_max_h}, for $\Sols$-geometric manifolds, there exists a unit Milnor basis $\{\bare_i\}$ with algebra
\begin{equation}\label{eq_algebra_bar_Sol}
	(\bar n^{ij})  = {\rm diag}(0,1,-1)\,, \\ \bar a_i = 0\,,
\end{equation}
such that the maximal Ricci tensor is
\begin{equation}\label{eq_Rbar_ij_bar_Sol}
	\sTbarRic = -2\, \bare^1 \otimes \bare^1\,.
\end{equation}

The (minimal) physical metric $\T h$ is a $\BVIn$ metric and its orthonormal Milnor basis $\{\e_i\}$ is characterized by
\begin{equation}\label{eq_algebra_Sol}
	(n^{ij})  = {\rm diag}(0,n_2,n_3)\,, \\ a_i = 0\,; \quad {\rm where} \ n_2>0 \ {\rm and} \ n_3 <0\,.
\end{equation}
When $n_2\not=-n_3$, we have $\Isom(\tilde\Sigma,\Tbarh) = Sol \rtimes D_2$ and the metric is maximal if and only if $n_2 = -n_3$. This means that a (minimal) $\BVIn$ metric and a maximal $\Sols$ metric have the same number of continuous symmetries, and their isometry group differ only by discrete symmetries. However, $\sTbarRic$ has more continuous symmetries than $\sTRic$ for $n_2\not=-n_3$, as ${\rm dim}( \Sym(\tilde\Sigma,\sTRic) ) = 3$ for $n_2\not=n_3$ while ${\rm dim}( \Sym(\tilde\Sigma,\sTbarRic) ) = \infty$ (see \Cref{app_Sol}).

The (constant) transformation matrix $\Lambda$ between $\{\bare_i\}$ and $\{\e_i\}$, with components defined in \eqref{eq:etoebartrafo}, can be put in the form
\begin{equation}\label{eq_Lambda_Sol}
    \Lambda =
    	 \begin{pmatrix}
            \sqrt{-n_2 n_3}  & 0 & 0 \\
            0 & \sqrt{-n_3} & 0 \\
            0 & 0 & \sqrt{n_2}
        \end{pmatrix}
        \begin{pmatrix}
            s\,  & c_1 & c_2 \\
            0 & c_4 & c_3 \\
            0 & s\, c_3 & s\, c_4
        \end{pmatrix},
\end{equation}
where $c_i \in \mR$ for $i = 1,2,3,4$, and $s = \pm 1$ with $(c_4)^2 \neq (c_3)^2$. The second matrix in~\eqref{eq_Lambda_Sol} is an automorphism of the Lie algebra~\eqref{eq_algebra_bar_Sol}. Because the components of $\sTbarRic$ in the $\{\bare_i\}$ basis in \eqref{eq_Rbar_ij_bar_Sol} are invariant by this automorphism, 
in the orthonormal basis $\{\e_i\}$ we get
\begin{equation}\label{eq_Rbarij_Sol}
    \sTbarRic = 2 \, n_2 n_3 \, \e^1 \otimes \e^1\,.
\end{equation}
Consequently, we have
\begin{equation}\label{eq_Rbar_ij_Sol_diff}
    \begin{split}
    \sTRic - \sTbarRic
    	&= -\frac{1}{2}\left(n_2 + n_3\right)^2 \e^1 \otimes \e^1
		+ \frac{1}{2}\left(n_2^2 - n_3^2\right) \e^2 \otimes \e^2
    		+ \frac{1}{2}\left(n_2^2 - n_3^2\right) \e^3 \otimes \e^3 \\
	&= \T 0\,; \quad {\rm iff} \ n_2 = -n_3\,,
	\end{split}
\end{equation}
leading to
\begin{equation}\label{eq_DeltaRs_Sol}
	\begin{split}
	{\rm tr}_{\T h}\left(\sTRic - \sTbarRic\right)
		&= -\frac{1}{2} \left(n_2 + n_3\right)^2 \\
		&\leq 0\,; \quad \forall \ n_2 > 0 \, {\rm and} \ n_3<0\\
		&=0\,; \quad {\rm iff} \ n_2 = -n_3\,.
	\end{split}
\end{equation}
Therefore, in general $\sTRic \not= \sTbarRic$, and the equality holds if and only if~$\T h$ is maximal.

\subsubsection{Field equations }

Similar to the Bianchi IX case, tilt is generally allowed in Bianchi VI$_0$ models. Therefore, the assumption $\T q = \T 0 = \T\tilt$ we take in this paper is a restriction in this model.

Using the orthonormal approach variables~\eqref{eq_set_orthonormal_no_tilt} and the expression of $\sTbarRic$ in the orthonormal basis~\eqref{eq_Rbarij_Sol}, 
the field equations of LSH solutions of topo-GR for non-tilted perfect fluid and non-tilted reference Ricci curvature, on a $\Sols$-geometric manifold read
\begin{subequations}\label{eq_syst_Sol_tot}
\begin{align}
	1
		&= \Omega_\sigma + \OmegaDR + \Omega_\rho\,,\label{eq_syst_Sol_1}\\
  	\dot H
		&= -(1+q)H^2\,, \label{eq_syst_Sol_2}\\
	\dot n_2
		&= - n_2 \left(H - 2\sigma_+ - 2\sqrt{3}\, \sigma_-\right), \label{eq_syst_Sol_3}\\
	\dot n_3
		&= - n_3\left(H - 2\sigma_+ + 2\sqrt{3}\, \sigma_- \right), \label{eq_syst_Sol_4}\\
	\dot\sigma_+
		&= -3H\sigma_+ - \frac{1}{6}\left(n_2 + n_3\right)^2, \label{eq_syst_Sol_5}\\
	\dot \sigma_-
		&= -3H\sigma_- - \frac{1}{2\sqrt{3}} \left(n_2^2 - n_3^3\right).\label{eq_syst_Sol_6}
\end{align}
\end{subequations}
The cosmological parameters take the form
\begin{equation}\label{eq_syst_Sol_Omega}
	H^2\Omega_\sigma
		= \sigma_+^2 + \sigma_-^2 \,, \\
	H^2\OmegaDR
		= \frac{1}{12} \left(n_2 + n_3\right)^2 \geq 0 \,.
\end{equation}

\subsection{\texorpdfstring{$\SLRRs$-geometry}{}}
\label{sec_SL2R}

\subsubsection{Reference curvature in the orthonormal basis}

As shown in Section~\ref{sec_max_Ric_max_h}, for $\SLRRs$-geometric manifolds, there exists a unit Milnor basis $\{\bare_i\}$ with algebra
\begin{equation}\label{eq_algebra_bar_SL2R}
	(\bar n^{ij})  = {\rm diag}(-1,1,1) \,, \\ {\bar a_i} = 0\,,
\end{equation}
such that the maximal Ricci tensor is
\begin{equation}\label{eq_Rbar_ij_bar_SL2R}
	\sTbarRic = 2 \left(\bare^1 \otimes \bare^1 - \bare^2 \otimes \bare^2 -  \bare^3 \otimes \bare^3\right).
\end{equation}

The (minimal) physical metric $\T h$ is a $\BVIII$ metric (Bianchi III models on $\SLRRs$ are subcases of Bianchi VIII models, see \Cref{sec_Isom_Bianchi}) and its orthonormal Milnor basis $\{\e_i\}$ is characterized~by
\begin{equation}\label{eq_algebra_SL2R}
	(n^{ij})  = {\rm diag}(n_1,n_2,n_3)\,, \\ a_i = 0\,; \\  n_2, n_3>0 \ {\rm and} \ n_1 <0\,.
\end{equation}
The metric is maximal if and only if $n_2=n_3$.

The (constant) transformation matrix $\Lambda$ between $\{\bare_i\}$ and $\{\e_i\}$ can be put in the form
\begin{align}\label{eq_Lambda_SL2R}
    \Lambda =
    	 \begin{pmatrix}
            \sqrt{n_2 n_3}  & 0 & 0 \\
            0 & \sqrt{-n_3 n_1} & 0 \\
            0 & 0 & \sqrt{-n_1 n_2}
        \end{pmatrix}
        A\,,
\end{align}
where $A$ is an automorphism of the Lie algebra~\eqref{eq_Lambda_SL2R}. The components~\eqref{eq_Rbar_ij_bar_SL2R} of the reference Ricci tensor are invariant under this automorphism. Therefore, in the orthonormal basis it reads
\begin{equation}\label{eq_Rbarij_SL2R}
    \sTbarRic = 2 \left( n_2 n_3 \, \e^1 \otimes \e^1 +  n_3 n_1 \, \e^2 \otimes \e^2 + n_1 n_2 \, \e^3 \otimes \e^3\right).
\end{equation}
Consequently, we have
\begin{equation}\label{eq_Rbar_ij_SL2R_diff}
\begin{aligned}
    \sTRic - \sTbarRic
    	&= \frac{\left(n_1+n_2+n_3\right)}{2} \left[
		    \left(n_1 - n_2 - n_3\right) \e^1 \otimes \e^1 + \left(n_2 - n_3 - n_1\right)\e^2 \otimes \e^2 \right. \\
		    &\qquad\qquad\qquad\qquad\qquad \left. + \left(n_3 - n_1 - n_2\right)\e^3 \otimes \e^3\right] \\
	    &= \T 0 \,; \quad {\rm iff} \ n_1 + n_2 + n_3 = 0 \,,
	\end{aligned}
\end{equation}
leading to
\begin{equation}\label{eq_DeltaRs_SL2R}
	\begin{aligned}
	{\rm tr}_{\T h}\left(\sTRic - \sTbarRic\right)
		&= -\frac{1}{2} \left(n_1 + n_2 + n_3\right)^2 \\
		&\leq 0\,,\\
		&=0\,; \quad {\rm iff} \ n_1 + n_2 + n_3 = 0\,.
	\end{aligned}
\end{equation}
Therefore, in general $\sTRic \not= \sTbarRic$. However, even if $\T h$ is maximal the equality does not necessarily hold, and reversely $\T h$ does not need to be maximal for the equality  to hold. This is the only geometry for which this happens.

\subsubsection{Field equations }

Similar to the Bianchi IX case, tilt is generally allowed in Bianchi VIII models. Therefore, the assumption $\T q = \T 0 = \T\tilt$ we take in this paper is a restriction in this model.

Using the orthonormal approach variables~\eqref{eq_set_orthonormal_no_tilt} and the expression of $\sTbarRic$ in the orthonormal basis~\eqref{eq_Rbarij_SL2R}, the field equations of LSH solutions of topo-GR for non-tilted perfect fluid and non-tilted reference Ricci curvature on a $\SLRR$-geometric manifold read
\begin{subequations}\label{eq_syst_SL2R_tot}
\begin{align}
	1
		&= \Omega_\sigma + \OmegaDR + \Omega_\rho\,,\label{eq_syst_SL2R_1}\\
  	\dot H
		&= -(1+q)H^2\,, \label{eq_syst_SL2R_2}\\
	\dot n_1
		&= - n_1 \left(H + 4\sigma_+ \right), \label{eq_syst_SL2R_3}\\
	\dot n_2
		&= - n_2 \left(H - 2\sigma_+ - 2\sqrt{3}\, \sigma_-\right), \label{eq_syst_SL2R_4} \\
	\dot n_3
		&= - n_3\left(H - 2\sigma_+ + 2\sqrt{3}\, \sigma_- \right), \label{eq_syst_SL2R_5} \\
	&\dot\sigma_+
		= -3H\sigma_+ + \frac{1}{6}\left(n_1+n_2+n_3\right)\left(2n_1-n_2-n_3\right), \label{eq_syst_SL2R_6} \\
	&\dot \sigma_-
		= -3H\sigma_- + \frac{1}{2\sqrt{3}}\left(n_1+n_2+n_3\right)\left(n_3-n_2\right). \label{eq_syst_SLRR_7}
\end{align}
\end{subequations}
The cosmological parameters take the form
\begin{equation}\label{eq_syst_SLRR_Omega}
	H^2\Omega_\sigma
		= \sigma_+^2 + \sigma_-^2\,, \\
	H^2\OmegaDR
		= \frac{1}{12} \left(n_1 + n_2 + n_3\right)^2 \geq 0\,.
\end{equation}

\subsection{$\RSS$-geometry}
\label{sec_RS2}

\subsubsection{Reference curvature in canonical coordinates}

As shown in Section~\ref{sec_max_Ric_max_h}, for $\RSS$-geometric manifolds, there exists a coordinate system $\{\bar x, \bar y, \bar z\}$ such that the maximal Ricci tensor is
\begin{equation} \label{eq_Rbar_KS}
	\sTRic = \T\dd \bar x \otimes\T\dd \bar x + \sin^2 \bar x \,\T\dd \bar y \otimes\T\dd \bar y\,.
\end{equation}

The (minimal) physical metric $\T h$ is a Kantowski--Sachs metric and there exists a coordinate system $\{x,y,z\}$ such that
\begin{equation} \label{eq_h_KS}
	\T h = A_1\left(\T\dd x \otimes\T\dd x + \sin^2 x \T\dd y \otimes\T\dd y\right) + A_2 \, \T\dd z \otimes\T\dd z\,,
\end{equation}
where $A_1, A_2 > 0$.

In general, $\{\bar x, \bar y, \bar z\}$ and $\{x,y,z\}$ can be different coordinate systems. However, because an isometry of the physical metric is necessarily a symmetry of $\sTbarRic$ as implied by the condition~\eqref{eq_GBKS_GBKS} (the minimal group for $\RSS$-geometries corresponds to the maximal group and is therefore the isometry group of $\T h$), then the two coordinate systems must coincide (see \Cref{app_RS2}). Consequently, we have $\sTRic = \sTbarRic$.

A similar proof for the $\RHH$-geometry could have been conducted.

\subsubsection{Field equations}

Similar to what was presented in \Cref{sec_RH2_syst} for the $\RHH$-geometry, using the properties derived in \Cref{app_RS2} and the momentum constraint~\eqref{eq_homo_Momentum}, we obtain $\sTbarRic\left(\T\tilt,\cdot\right) = \T 0$, $\div_{\T h} \, \T\sigma = \T 0$, and $\T q = \T 0$. Therefore, as for GR, LSH solutions on $\RSS$-geometric manifolds do not allow for tilted fluids. No reference tilt is allowed either. As a result, the field equations are equivalent to the ones for $\RHH$-geometric manifolds, corresponding to~\eqref{eq_syst_RH2_tot}. Therefore, LSH solutions on $\RHH$ and $\RSS$-geometric manifolds are dynamically equivalent and only differ by their spatial geometry.

\subsection{General results revisited}
\label{sec_discussion_exact}

In \Cref{sec_LSH_General_results}, we derived general results for LSH solutions without specifying the precise model. We discussed, in particular, shear-free and static vacuum solutions, and conditions for late-time isotropization. The computation of the field equations for (non-tilted) perfect fluids in \Cref{sec_E3} to \Cref{sec_RS2} for each model allows us to refine these results. Notably, for some models, these equations allowed us to determine the generic behavior of the shear parameter $\Omega_\sigma$, the sign of the modified curvature parameter $\OmegaDR$ and whether or not tilt is allowed. We summarize this in Table~\ref{tab_dynamics}. Additionally, these field equations allowed us to obtain a sufficient condition for an LSH solution of topo-GR to be static and vacuum:
\begin{proposition}
    Consider an LSH solution of topo-GR with a non-tilted perfect fluid fulfilling the weak energy condition. For all geometries apart from the non-LRS Bianchi II models in the $\Nils$-geometry, if $H=0$ at a given time~$t_0$, then $H=0$, $\T\sigma = \T 0$, $\sTRic = \sTbarRic$, and $\rho =0$ hold for all time~$t$, and the solution is a static vacuum solution.
\end{proposition}
This result shows that bounces and recollapses are \emph{not} possible in topo-GR for non-tilted perfect fluids with weak energy condition for all geometries except for the non-LRS Bianchi II model in the $\Nils$-geometry.

The isotropization \Cref{thm_isotropization} requires $\sR-\sbarR \leq 0$ to hold. As shown in \Cref{sec_E3} to \Cref{sec_RS2}, and summarized in \Cref{tab_R_Rbar}, this is the case for all geometries apart from the non-LRS $\Nils$-geometry, leading the following proposition.
\begin{proposition}\label[theorem]{thm_isotropization_result}
	 Consider an LSH solution of topo-GR with a positive cosmological constant, matter satisfying the dominant and strong energy conditions, and $H>0$ at an initial time~$t_0$. Then for all geometries apart from the non-LRS Bianchi II models in the $\Nils$-geometry, we have $\T\sigma \overset{t\rightarrow\infty}{\longrightarrow} \T 0$. Moreover, for non-tilted perfect fluids, we have $\sTRic - \sTbarRic \overset{t\rightarrow\infty}{\longrightarrow} \T 0$.
\end{proposition}
Therefore, unless we consider a non-LRS Bianchi II model, regardless of the initial conditions, a sufficiently long phase of inflation guarantees, within topo-GR, that the Universe becomes shear-free at late-times. This holds in particular for the two geometries that allow for positive scalar curvature for the class of LSH solutions, i.e., $\mS^3$ and $\RSS$-geometries, hence contrasting with GR. However, this isotropization is not ensured for non-LRS Bianchi II models. We discuss this result in \Cref{sec_discussion_Nil}.

Note that for geometries where $\sR\leq0$ for which Wald's theorem holds in GR, due to the unknown sign of $\sbarR$, it was \emph{a priori} not guaranteed that $\sR-\sbarR \leq 0$ would still hold for these geometries. In this sense, it is quite remarkable that Wald's theorem is still valid in topo-GR for almost all cases.

\begin{table}[t]
\centering\small
\newcommand{\ToBeDet}{?}
\caption{
	\small $\Omega$-parameters and tilt of LSH solutions of topo-GR. A ``\,\ToBeDet\,''  refers to a more complex dynamical behavior of the quantity, which we intend to study in future works.\label{tab_dynamics}
}
\renewcommand{\arraystretch}{2}
\begin{tabular}{llccc}

\toprule 
\makecell[l]{{\bf Maximal }\\{\bf geometry}}
& \makecell[l]{{\bf Minimal }\\{\bf metric}}
& \centering $H^2 \Omega_\sigma$
& \centering $H^2\OmegaDR$
& \centering \bf Allows tilt? \arraybackslash
\\ \midrule
    
$\mE^3$
& $\begin{dcases} \BI \\ \BVIIn \end{dcases}$
& $\begin{dcases} \propto 1/\sfac^6 \\ \text{\ToBeDet} \end{dcases}$
& $\begin{dcases} 0 \\ \geq 0 \end{dcases}$
& $\begin{dcases} {\rm No} \\ {\rm Yes} \end{dcases}$
\\

$\mS^3$
& $\BIX$
& \ToBeDet
& $\geq 0$
& Yes
\\

$\mH^3$
& $\BV$ (or $\BVIIh$)
& $0$
& $0$
& No
\\

$\mR\times\mH^2$
& $\BIII$
& $\propto 1/\sfac^6$
& $0$
& No
\\

$\Sols$
& $\BVIn$
& \ToBeDet
& $\geq 0$
& Yes
\\

$\Nils$
& $\BII$
& \ToBeDet
& $\begin{dcases} \ \ \ \ \ \; \, {\rm LRS:} \geq 0\\  {\rm non\text{-}LRS:~indefinite} \end{dcases}$
& $\begin{dcases} \ \ \ \ \ \: \, {\rm LRS:~No} \\ {\rm non\text{-}LRS:~Yes} \end{dcases}$
\\

$\SLRRs$
& $\BVIII$ (or $\BIII$)
& \ToBeDet
& $\geq 0$
& Yes
\\
    
$\RSS$
& KS
& $\propto 1/\sfac^6$
& $0$
& No
\\
\bottomrule

\end{tabular}
\end{table}

\section{To what extent is the result for the $\Nils$-geometry peculiar?}
\label{sec_discussion_Nil}

Because ${\rm tr}_{\T h}\left(\sTRic - \sTbarRic\right) \leq 0$ fails for a general LSH solution on the $\Nils$-geometry, the isotropization theorem does not apply for that geometry, unless the LRS condition is assumed. This is a very peculiar and unexpected result. Indeed, the $\Nils$-geometry is, algebraically speaking, one of the simplest geometries, and, in GR, LSH solutions on this geometry (i.e., Bianchi II solutions) are, dynamically speaking, the second simplest LSH solutions after the Bianchi I solutions. In particular, Wald's theorem holds for Bianchi II solutions in GR. In this sense, we did not expect any fundamental difference for the Bianchi II solutions between the two theories. 
Moreover, four purely geometric properties make LSH solutions of topo-GR on the $\Nils$-geometry even more~peculiar:
\begin{enumerate}
    \item for all geometries apart from the non-LRS $\Nils$-geometry, when the metric $\T h$ has more local isometries than that imposed by the minimal group, these isometries are always symmetries of $\sTbarRic$, i.e., $\Isom(\tilde\Sigma,\T h) \subseteq \Sym(\tilde\Sigma,\sTbarRic)$ always holds. For the $\Nils$-geometry, this is not the case, as the fourth Killing vector field of the metric is not necessarily a collineation vector field for $\sTbarRic$, unless $r_2 = r_3$ in \eqref{eq_Rbar_ij_Nil}.
    
    \item for all geometries except for the $\Nils$-geometry, when $\sTRic$ is maximal, then $\sTRic = \sTbarRic$. For the $\{\mE^3,\ \mH^3,\ \RHH,\ \RSS\}$-geometries this is trivially obtained. For the $\{\mS^3, \ \Sols, \ \SLRRs\}$-geometries, maximality of $\sTRic$ is achieved when $n_1=n_2=n_3$, $n_2 = -n_3$ and $n_1 + n_2 + n_3 = 0$, respectively, and relations \eqref{eq_Rbar_ij_S3_diff}, \eqref{eq_Rbar_ij_Sol_diff} and \eqref{eq_Rbar_ij_SL2R_diff} imply $\sTRic = \sTbarRic$. But, the $\Nils$-geometry is the only geometry where this property does not hold, as $\sTRic$ is always maximal but $\sTRic \not= \sTbarRic$ by relation \eqref{eq_Rbarij_Nil_diff}, in general.
    
    \item for all geometries other than the $\Nils$-geometry, given a basis in which $\T h$ is known, the components of $\sTbarRic$ are fully expressed as functions of the components of $\T h$ and the commutation coefficients of the basis. Taking the example of the orthonormal Milnor basis for Bianchi metrics, then $\sTbarRic$ is uniquely expressed for each model as a function of $n_i$ and $a$. However, for the $\Nils$-geometry, two additional parameters are present, i.e., $r_2$ and $r_3$.
    
    \item  the $\Nils$-geometry is the only geometry where the absence of reference tilt $\T\tilt$ is needed to get $\div_{\T h} \sTbarRic = \T 0$ via \eqref{eq_homo_tilt_dot}. Additionally, it is the only geometry for which the evolution equation~\eqref{eq_homo_Rbar_dot} for $\sTbarRic$ provides equations independent of the other 3+1-equations.
\end{enumerate}

We were unable to identify a definite explanation for the mentioned peculiarity of the $\Nils$-geometry. We hoped to find a ``hidden'' geometric or topological condition that would impose either $\Isom(\tilde\Sigma,\T h) \subseteq \Sym(\tilde\Sigma,\sTbarRic)$ or directly $\sTRic = \sTbarRic$. When the topology is $T^3(n)/\mathbb{Z}_{k=3,4,5}$ (see \cite{1998_Kodama,2002_Kodama} for the notation and the proof) LRS is imposed, and therefore $\Isom(\tilde\Sigma,\T h) \subseteq \Sym(\tilde\Sigma,\sTbarRic)$ (as shown in \Cref{sec_Nil_LRS}). But for the other closed manifolds modeled on $\Nils$, no such constraint arises, and the peculiarities listed above remain.

We believe that these peculiarities, and especially the first one in the list above, suggest that the definition of $\stTbarRic$ should be slightly changed such that either $\Isom(\tilde\Sigma,\T h) \subseteq \Sym(\tilde\Sigma,\sTbarRic)$ always holds, or perhaps the stronger condition ``$\sTRic$ is maximal $\Rightarrow$ $\sTRic = \sTbarRic$'' always holds. That is, we would eliminate by definition (part of, or all of) the diffeomorphism freedom that might remain between $\sTRic$ and $\sTbarRic$ when both are maximal. Ideally, this redefinition should imply $\Isom(\tilde\Sigma, \T h) \subseteq \Sym(\tilde\Sigma, \sTbarRic)$ prior to considering the field equation. We recall that in Section~\ref{sec_Program}, we had to use the field equations to show that the Killing vectors related to the $\Ghatmin \subseteq \Isom(\tilde\Sigma,\T h)$ of the spatial metric are collineations of $\sTbarRic$, leading to \eqref{eq_GBKS_GBKS}.

Imposing $\Isom(\tilde\Sigma,\T h) \subseteq \Sym(\tilde\Sigma,\sTbarRic)$ is, however, foliation dependent, as it features explicitly the spatial metric $\T h$. We would rather want this property to hold regardless of the chosen foliation. Therefore, the ``redefinition'' of $\stTbarRic$ would be made such that, as a consequence, the following property would hold.

\begin{property}\label[property]{property_redefinition}
For \emph{all} $\Sigma$-foliations, spacelike with respect to $\T g$, $\Isom(\tilde\Sigma, \T h) \subseteq \Sym(\tilde\Sigma, \sTbarRic)$ holds, where $\T h$ and $\sTbarRic$ are, respectively, the spatial metric induced by $\T g$ on the foliation and the spatial reference Ricci tensor induced by $\stTbarRic$ on the foliation.
\end{property}
Adding \Cref{property_redefinition} to the definition of topo-GR would be very artificial. Ideally, we would want that the way $\stTbarRic$ is constructed would encompass both the above property, and the current definition~\eqref{eq_def_nabla_bar} in terms of external Whitney sum. Even more ideally, we would like this definition not to depend on the spacetime manifold being topologically $\mR\times\Sigma$ as assumed in the paper (see the discussion in \Cref{sec_topo_of_M}). However, we have not yet found a natural and compact way of doing that.

Finally, let us clarify an important point here: the result for $\Nils$ and the discussion we just made do \emph{not} mean that the theory is currently ill-defined. Indeed, the system of equations~\eqref{eq_syst_Nil_tot} for $\Nils$ is closed, and more generally, the 3+1-system derived in \Cref{sec_3+1} is closed. Rather, the results for non-LRS Bianchi II models on the $\Nils$-geometry contradict the understanding and interpretation we have had so far of the theory, namely that the solutions of the theory are qualitatively the same regardless of which topology is chosen. With this respect, if one wants the topo-GR theory to fully comply with this line of interpretation, one might want   \Cref{property_redefinition} to result from the definition of the theory.

\section{Conclusion }
\label{sec_Conclusion}

In this work, we have determined the systems of equations for locally spatially homogeneous (LSH) solutions, admitting a Bianchi--Kantowski--Sachs (BKS) metric, of a topology dependent modified gravity theory, \emph{topo-GR}. We focused on non-tilted LSH solutions for closed $3$-manifolds, hence using the correspondence between the Thurston classification and BKS metrics. Except for the $\Nils$-geometry, all models feature the same number of free-parameters as in general relativity (GR). Our four main results are:
\begin{enumerate}
	\item All BKS models admit shear-free exact solutions for perfect fluids with a dynamics equivalent to that of a flat FLRW metric in GR. This contrasts with GR, for which for maximal geometries other than $\{\mE^3, \ \mS^3, \ \mH^3\}$, shear-free solutions are only possible if the matter has a non-vanishing anisotropic stress (cf., e.g., \cite{1993_Mimoso}).
	
	\item All BKS models admit a static vacuum solution. The spacetime metric features no lapse and the induced spatial metric is static and corresponds to the maximal metric (apart from the $\SLRRs$-geometry where it can be minimal) related to the spatial topology. In contrast, static vacuum LSH solutions in GR are only possible for the $\mE^3$-geometry, with the Minkowski metric.
	
	\item For all BKS models, apart from the non-locally rotationally symmetric (non-LRS) Bianchi II model, recollapse is \emph{never} possible when the weak energy condition is fulfilled. In contrast, Bianchi IX solutions ($\mS^3$-geometry) and Kantowski-Sachs solutions ($\RSS$-geometry) can recollapse in GR (cf., e.g., \cite{1988_Barrow}).
	
	\item For all BKS models, apart from the non-LRS Bianchi II model, a Wald-like theorem~\cite{1983_Wald} holds: a positive cosmological constant ensures isotropization at late times. This differs from GR for which LSH solutions on the $\mS^3$ and $\RSS$ geometries do not isotropize without fine-tuning on the value of the initial spatial curvature.
\end{enumerate}
With respect to the recollapse and the isotropization problems, these results show that topo-GR behaves better than GR without introducing additional parameters, with the only exception of the non-LRS Bianchi II model. We discussed in \Cref{sec_discussion_Nil} why this latter result is peculiar.

Although the field equations (and the Lagrangian) of topo-GR explicitly depend on the topology of the Universe, the dependence of their solutions on topology, however, is actually greatly reduced compared to GR. A typical example, already discussed in \cite{2023_Vigneron_et_al_b, 2024_Vigneron_et_al_b}, is the expansion law of homogeneous-isotropic solutions: in GR it differs depending on whether the topology is Euclidean, spherical, or hyperbolic, whereas it is identical in topo-GR. The results of this paper further support that interpretation for all topologies admitting LSH solutions. In this sense, the physics described by topo-GR can be said to be \emph{universal} with respect to topology.\\

In direct continuation of this work, there are two main perspectives: dynamical analysis and tilted models. Indeed, while we provided some exact solutions for each model, we did not conduct an in-depth analysis of the systems of equations. In this regard, studying the dynamics towards the initial singularity and the presence or absence of chaos will be one of the main follow-ups. Further, while for the $\{\mE^3 \, ({\rm with \ Bianchi \ I}), \ \mH^3, \ \RSS, \ \RHH\}$-geometries generic solutions do not feature tilt, for the other geometries, we imposed a vanishing tilt. Generalizing this work to tilted models for the $\{ \mS^3, \Nils, \Sols, \SLRRs\}$-geometries would be therefore another main follow-up.

Finally, as discussed in Section~\ref{sec_shear_free}, the existence of static vacuum solutions in any topology could be beneficial for quantizing inflation and defining a Bunch--Davies vacuum in any topology, as was shown in \cite{2024_Vigneron_et_al_b} for $\mS^3$ and $\mH^3$ types of topologies. Alongside this potential benefit for inflation, shear-free solutions in topo-GR also suggest the possibility to extend the $\Lambda$CDM model to any Thurston geometry without adding any parameter (other than the spatial curvature already present in the $\Lambda$CDM model) and without changing the background dynamics. In this sense, topo-GR would offer a way to define a universal cosmological model, formally equivalent for any spatial topology.

\section*{Acknowledgements}

The authors are grateful to Abdelghani Zeghib and \' Aron Szab\' o for very insightful discussions. QV was supported in part by the European Research Council (ERC) under the European Union's Horizon 2020 research and innovation program (grant agreement ERC advanced grant 740021-ARTHUS, PI: T.B.).
QV is additionally supported by the Polish National Science Centre under Grant No. SONATINA 2022/44/C/ST9/00078. 
HB was funded by the Austrian Science Fund (FWF) [Grant DOI: 10.55776/J4803].
For open access purposes, the authors have applied a CC BY public copyright license to any author-accepted manuscript version arising from this submission.

\appendix

\section{Some properties of locally homogeneous quantities}

\begin{table}[t]
    \centering\small
    \caption{\small List of groups \label{tab_def_Var}}
    \begin{tabular}{ll}

        \toprule
        \bf \makecell[l]{\bf Group} & \bf Description\arraybackslash \\
        \midrule
        $\rmD_{2n}$     & Dihedral group of rank $n$ \\
        $\rmO(n)$       & Orthogonal group of dimension $n$\\
        $\rmO_+(n,m)$   & \makecell[l]{Positive orientation preserving component of the indefinite \\ \ \ \ orthogonal group of dimension $n$ (positive) and $m$ (negative)}\\
        $\IO(n)$        & Euclidean group of dimension $n$\\
        $\Nil$          & Heisenberg group\\
        $\Sol$          & Sol group \\
        $\SLRR$         & Universal cover of the ${\rm SL}(2,\mR)$ group\\
        ${\rm G}_{\blacklozenge}$  & Maximal group associated with the $\blacklozenge$-geometry\\
        ${\rm B}_{\blacklozenge}$ & Bianchi $\blacklozenge$ group \\ 
        \bottomrule
    \end{tabular}
\end{table}

\subsection{$\mH^3$-geometry}\label{app_H3}

\begin{proposition}\label[proposition]{prop_property_H3}
    Consider a closed manifold $\Sigma$ modeled on the $\mH^3$-geometry. Then,
    \begin{enumerate}
        \item there are no non-trivial vector fields that are locally invariant by the $\BV$ or $\BVIIn$ groups,
        \item all symmetric (0,2)-tensor fields locally invariant by the $\BV$ or $\BVIIn$ groups are homothetic.
    \end{enumerate}
\end{proposition}

\begin{proof}

We first consider the $\BV$ group. Let $\T h$, $\T v$, and $\T A$ be, respectively, a Riemannian metric, a vector field and a symmetric (0,2)-tensor on $\Sigma$ all (locally) left-invariant by the $\BV$ group. The metric $\T h$ is any arbitrary left-invariant metric introduced for the proof. Then, by compactness, $\div_{\T h} \T v = 0$. Additionally, regardless of the metric, there always exists an orthonormal Milnor basis $\{\e_i\}$ such that $n^{ij} = 0$ and $(a_i) = (a,0,0)$. Computing $\div_{\T h} \T v$ in that basis yields $v^1 = 0$. Then computing $\T\nabla_{\T v} \T v$ in that basis yields $(a\,(v^2)^2 + a\,(v^3)^2,0,0)$. Since the first component needs to be zero, we get $v^2 = v^3 = 0$. Therefore, any left-invariant vector field is zero. Consequently, both $\nabla^i A_{ij}$ and $\nabla^j\left(A_{jk} A^{ki}\right)$ are zero. Computing these quantities in the Milnor basis directly yields $\T A \propto \T h$. Therefore, all left-invariant symmetric (0,2)-tensors are homothetic to a left-invariant metric.

We now consider the $\BVIIn$ group. An orthonormal Milnor basis is characterized by $(n^{ij}) = {\rm diag}(0,n_2, a^2/(h\, n_2))$ and $(a_i) = (a,0,0)$ with $n_2>0$. Because any locally homogeneous metric on a closed hyperbolic manifold is isotropic \cite{1968_Mostow}, then $\sTRic \propto \T h$, which leads to $n_2 = a/\sqrt{h}$. Then, conducting the same calculation as for the Bianchi V group leads to $\T v = 0$ and $\T A \propto \T h$.
\end{proof}

\begin{table}[t!]
\centering\small
\caption{\small Generators of left-translations $\{\T\xi_i\}$, left-invariant vector basis $\{\e_i\}$ and its dual basis $\{\e^i\}$ for the Bianchi models considered in this appendix and elsewhere. These quantities are used in local coordinates, resulting from the Maurer--Cartan equation~\eqref{eq_Maurer_Cartan}. $\{\T\xi_i\}$ and $\{\e_i\}$ are chosen such that their algebra are those given in \Cref{tab_Bianchi_groups} with unit structure constants.\label{tab_Bianchi_vectors}}
\renewcommand{\arraystretch}{1.}
\newcommand{\vspacetable}[1]{\multicolumn{2}{c}{} \vspace{#1}\\}
\setlength{\tabcolsep}{5pt}
\begin{tabular}{llll}

\toprule
\makecell[l]{\bf Bianchi \\ \bf groups}
& $\T\xi_i$
& $\e_i$
& $\e^i$
\\
\midrule


$\BII$
& $\begin{aligned}
    &\Tparx \\
    &-\tfrac{z}{2}\Tparx + \Tpary \\
    &\tfrac{y}{2}\Tparx + \Tparz
\end{aligned}$
& $\begin{aligned}
    &-\Tparx \\
    &\tfrac{z}{2}\Tparx + \Tpary \\
    &-\tfrac{y}{2}\Tparx + \Tparz
\end{aligned}$
& $\begin{aligned}
    &-\T\dd x  + \tfrac{z}{2} \T\dd y - \tfrac{y}{2} \T\dd z\\
    &\T\dd y \\
    &\T\dd z
\end{aligned}$
\\
\vspacetable{-12pt}
\hline
\vspacetable{-12pt}



$\BVIn$
& $\begin{aligned}
	&\T\partial_x + z\,\T\partial_y + y\,\T\partial_z \\
	&\sinh x\, \T\partial_y + \cosh x\, \T\partial_z \\
	&-\cosh x\,\T\partial_y - \sinh x\,\T\partial_z
\end{aligned}$
& $\begin{aligned}
	&\T\partial_x \\
	&\sinh x\, \T\partial_y + \cosh x\, \T\partial_z \\
	&-\cosh x\,\T\partial_y - \sinh x\,\T\partial_z
\end{aligned}$
& $\begin{aligned}
	&\T\dd x \\
	&-\sinh x\, \T\dd y + \cosh x\, \T\dd z \\
	&-\cosh x\,\T\dd y + \sinh x\,\T\dd z
\end{aligned}$
\\
\vspacetable{-12pt}
\hline
\vspacetable{-12pt}


& $\begin{aligned}
	&\frac{\cosh z}{\cosh y}\T\partial_x + \sinh z\,\T\partial_y \\
	    &\quad- \tanh y \cosh z\,\T\partial_z
\end{aligned}$
& $\begin{aligned}
	&-\T\partial_x
\end{aligned}$
& $\begin{aligned}
	& -\T\dd x + \sinh y\,\T\dd z
\end{aligned}$ \\
\vspacetable{-12pt}
\cline{2-4}
\vspacetable{-10pt}
$\BVIII$
&$\begin{aligned}
	&\frac{\sinh z}{\cosh y}\T\partial_x + \cosh z\,\T\partial_y \\
	    &\quad- \tanh y \sinh z\,\T\partial_z
\end{aligned}$
&$\begin{aligned}
	&\cos x \tanh y\, \T\partial_x + \sin x\,\T\partial_y \\
	    &\quad + \frac{\cos x}{\cosh y} \T\partial_z
\end{aligned}$
&$\begin{aligned}
	& \sin x\, \T\dd y + \cos x\cosh y\,\T\dd z
\end{aligned}$ \\
\vspacetable{-12pt}
\cline{2-4}
\vspacetable{-12pt}
&$\begin{aligned}
	&-\T\partial_z
\end{aligned}$
&$\begin{aligned}
	&-\sin x \tanh y\, \T\partial_x + \cos x\,\T\partial_y \\
	    &\quad - \frac{\sin x}{\cosh y} \T\partial_z
\end{aligned}$
&$\begin{aligned}
	& \cos x\, \T\dd y - \sin x\cosh y\, \T\dd z
\end{aligned}$
\\


\bottomrule
\end{tabular}
\end{table}

\subsection{$\RHH$-geometry}
\label{app_RH2}

\begin{proposition}\label[proposition]{prop_property_RH2}
    Let $\Sigma$ be a closed manifold modeled on the $\RHH$-geometry. Let $\T v$ and $\T A$ be, respectively, a vector field and a symmetric $(0,2)$-tensor on $\Sigma$, both locally invariant by a $\BIII$ group. Then, there exists (locally) a Milnor basis $\{\e_i\}$ with $(n^{ij}) = {\rm diag}(0,a,-a)$ and $(a_i) = (a,0,0)$, such that
    \begin{align}
        \T v = \left(0,b, b \right), \label{eq_form_v_BIII}
    \end{align}
    where $b$ is a constant, and
    \begin{align}\label{eq_TTensor_RH2_Compact}
        \T A = b_1 \, \T \delta +
        \begin{pmatrix}
            -2 b_2 & 0 & 0\\
            0 & b_2 & 3 b_2\\
            0 & 3 b_2 & b_2
        \end{pmatrix},
    \end{align}
    where $\T \delta$ is the Kronecker delta tensor, and $b_1$ and $b_2$ are constants.
\end{proposition}

\begin{proof}
Let $\T h$, $\T v$ and $\T A$ be, respectively, a Riemannian metric, a vector field, and a symmetric (0,2)-tensor on $\Sigma$, all (locally) left-invariant by $\BIII$. The metric $\T h$ is any arbitrary left-invariant metric introduced for the proof. Then, there exists an orthonormal Milnor basis $\{\e_i\}$ such that $(n^{ij}) = {\rm diag}(0,a,-a)$ and $(a_i) = (a,0,0)$, where compactness imposes $n_2 = a$ (see \Cref{sec_Isom_Bianchi}). Then, following a similar method to the proof of \Cref{prop_property_H3}, we obtain \eqref{eq_form_v_BIII} and \eqref{eq_TTensor_RH2_Compact}. Note that $\T v$ is an eigenvector of the Ricci tensor of $\T h$ with eigenvalue zero.
\end{proof}

\subsection{$\Nils$-geometry}
\label{app_Nil}

For any $\BII$ Riemannian metric, there exists a unit Milnor basis with algebra given in \Cref{tab_Bianchi_groups} and such that
\begin{equation}\label{eq_app_Nil_metric}
	\T{h} = A_1 \e^1 \otimes \e^1 + \e^2 \otimes \e^2 + \e^3 \otimes \e^3\,,
\end{equation}
where $\mR \ni A_1 > 0$. Any such metric is maximal, i.e., $\Isom(\tilde\Sigma,\T h) \cong \GNil$ with ${\dim(\GNil) = 4}$, hence, admitting a fourth Killing vector
\begin{equation}
    \T\xi_4 \coloneqq z\,\Tpary - y\, \Tparz\,.
\end{equation}
\begin{proposition}\label[proposition]{prop_property_Nil}
    Let $\T h$ and $\T A$ be, respectively, a Riemannian metric and a symmetric (0,2)-tensor, both left-invariant by the $\BII$ group. If $\Isom(\tilde\Sigma,\T h) \subseteq \Isom(\tilde\Sigma,\T A)$, then
    \begin{equation*}
        \T A = A_{11} \e^1 \otimes \e^1 + A_{22} \left(\e^2 \otimes \e^2 + \e^3 \otimes \e^3\right),
    \end{equation*}
    where $A_{11}, A_{22} \in \mR$, in the unit Milnor basis for which the metric has the form \eqref{eq_app_Nil_metric}.
\end{proposition}

\begin{proof}
    A direct computation of $\Lie{\T\xi_4} \T A = 0$ for a symmetric tensor $\T A = A_{ij} \e^i\otimes\e^j$ with $A_{ij}$ being constants, using the coordinate expressions for the dual basis $\{\e^i\}$ in \Cref{tab_Bianchi_vectors}, leads to the result.
\end{proof}
\begin{remark}
    While all $\BII$ Riemannian metrics can be put into the single form \eqref{eq_app_Nil_metric}, there are three different ways of defining a $\BII$ Lorentzian metric (see, e.g., \cite{1992_Rahmani}). Interestingly, a Lorentzian metric $\T h$ for which the center is a null vector field, i.e., $[\T \xi_3, \T \xi_i] = 0$ for $i = 2,3$ and $\T h(\T\xi_3,\T\xi_3)=0$, is Ricci flat. Riemannian $\BII$ metrics cannot be flat.
\end{remark}
%

\subsection{$\Sols$-geometry}
\label{app_Sol}

For any $\BVIn$ metric $\T h$, there always exists a unit Milnor basis $\{\e_i\}$ such that 
\begin{equation}\label{eq_ap_metric_VI}
	\T h = A_1 \,\e^1 \otimes \e^1 + A_2 \,\e^2 \otimes \e^2 + A_3 \,\e^3 \otimes \e^3\,,
\end{equation}
where $A_i > 0$ for $i=1,2,3$.\footnote{The automorphism $\{\e_1,\e_2,\e_3\} \rightarrow \{\e_1,\alpha \e_2,\alpha\e_3\}$, which keeps the algebra of the basis, allows us to further choose the value of either $A_2$ or $A_3$.} The maximal metrics are those with $A_2 = A_3$.

\begin{proposition}\label[proposition]{prop_property_Sol}
Let $\T h$ be a $\BVIn$ metric, then
\begin{align*}
    \aff(\mR^3, \T h)
        &= {\rm span}\left\{\T\xi_1, \T\xi_2, \T\xi_3\right\}\,; \quad \forall A_i > 0\,,\\
    \colRic(\mR^3, \T h)
        &=
        \begin{dcases}
            {\rm span}\left\{\T\xi_1, \T\xi_2, \T\xi_3\right\}\,; & A_2 \not=A_3 \,,\\
            \left\{\T v \in \mathfrak{X}(\mR^3)|\ \T{v} = v^1 \Tparx + v^2(x,y,z) \Tpary + v^3(x,y,z) \Tparz\right\}\,; &  A_2 =A_3\,,
        \end{dcases}
\end{align*}
where $\left\{\T\xi_1, \T\xi_2, \T\xi_3\right\}$ and $A_i$ are defined in \Cref{tab_Bianchi_vectors} and \eqref{eq_ap_metric_VI}, respectively, and where $v^1$ is a constant, and $v^2(x,y,z)$ and $v^3(x,y,z)$ are arbitrary functions. 
\end{proposition}

\begin{remark}
    The metric always has $3$ linearly independent affine collineations. If the metric is minimal, i.e., $A_2\not= A_3$ there are only $3$ linearly independent Ricci collineations as well. However, when the metric is maximal, i.e., $A_2=A_3$, the set of Ricci collineations is infinite dimensional. The infinite dimensional part lies in the kernel of the Ricci tensor.
\end{remark}

\begin{proof}
The Ricci tensor of~\eqref{eq_ap_metric_VI} is
\begin{equation}\label{eq_ap_Ric_proof_Sol}
\begin{aligned}
	\sTRic[\T h] &=
		-\frac{(A_2+A_3)^2}{2 A_2 A_3} \,\e^1 \otimes \e^1 + \frac{(A_2 - A_3)(A_2+A_3)}{2 A_3 A_1} \,\e^2 \otimes \e^2 \\
		&\quad - \frac{(A_2 - A_3)(A_2+A_3)}{2 A_1 A_2} \,\e^3 \otimes \e^3\,.
\end{aligned}
\end{equation}
For $A_2 \not=A_3$, the Ricci tensor is always non-degenerate, and therefore it is a (semi-Riemannian) metric. Its Ricci tensor is therefore well-defined and given by
\begin{equation}\label{eq_ap_Ric_Ric_proof_Sol}
\begin{aligned}
	\sTRic\left[\sTRic[\T h]\right] &=
		\frac{(A_2-A_3)^2}{2 A_2 A_3} \,\e^1 \otimes \e^1 + \frac{(A_2 - A_3)^2}{2 A_3 A_1} \,\e^2 \otimes \e^2 \\
		&\quad + \frac{(A_2 - A_3)^2}{2 A_1 A_2} \,\e^3 \otimes \e^3\,.
\end{aligned}
\end{equation}
For $A_2 \not=A_3$, this tensor is always non-degenerate with positive eigenvalues, and therefore it is a Riemannian $\BVIn$ metric. Therefore, $\sym(\mR^3,\TRic[\TRic[\T h]]) = \colRic(\mR^3,\TRic[\T h]) = {\rm span}\left\{\T\xi_1, \T\xi_2, \T\xi_3\right\}$. By the inclusion 
\begin{equation*}
    \sym\left(\mR^3,\T h\right)
        \subseteq \aff\left(\mR^3,\T h\right)
        \subseteq \colRic\left(\mR^3,\T h\right)
        \subseteq \sym\left(\mR^3,\TRic\left[\TRic\left[\T h\right]\right]\right),
\end{equation*}
this implies 
\begin{equation}\label{eq:affRiceq}
    \aff(\mR^3,\T h) = \colRic(\mR^3,\T h) = {\rm span}\left\{\T\xi_1, \T\xi_2, \T\xi_3\right\}.
\end{equation}
Assuming $A_2 =A_3$, we find
\begin{equation}
\label{eq_ap_Ric_proof_Sol_A2_A3}
	\sTRic[\T h] = -2 \,\e^1 \otimes \e^1\,.
\end{equation}
Let $\T v = v^1(x,y,z) \Tparx + v^2(x,y,z) \Tpary + v^3(x,y,z) \Tparz$ be a vector field such that $\Lie{\T v} \sTRic[\T h] = 0$. This leads to $\partial_i v^1 = 0$ for $i \in \{x,y,z\}$, i.e.,
\begin{equation}\label{eq_app_col_Ric_proof_Sol}
	\T v = v^1 \Tparx + v^2(x,y,z) \Tpary + v^3(x,y,z) \Tparz \,.
\end{equation}
Finally, because an affine collineation vector field is always a Ricci collineation, it must have this form. Then, by solving $\Lie{\T v} \T\nabla[\T h] = 0$ with $A_2 = A_3$ for \eqref{eq_app_col_Ric_proof_Sol} being the affine collineation, it directly follows that the only solutions are the vectors~$\{\T\xi_1,\T\xi_2,\T\xi_3\}$.
\end{proof}
%

\subsection{$\SLRRs$-geometry}
\label{app_SL2R}

\subsubsection{Maximal Ricci tensor}

For any $\BVIII$ metric $\T h$, there always exists a unit Milnor basis $\{\e_i\}$ such that
\begin{align}\label{eq_ap_metric_VIII}
	\T h = A_1 \e^1 \otimes \e^1 + A_2 \e^2 \otimes \e^2 + A_3 \e^3 \otimes \e^3\,,
\end{align}
where $A_i > 0$ for $i=1,2,3$. The Ricci tensor in this basis reads
\begin{align}\label{eq_app_SL2R_Ric_min}
	\begin{split}
	\sTRic[\T h] &=
		\frac{(A_1+A_2-A_3)(A_1-A_2+A_3)}{2 A_2 A_3} \,\e^1 \otimes \e^1 \\
		&\quad + \frac{(A_1- A_2 + A_3)(A_1+A_2+A_3)}{2 A_3 A_1} \,\e^2 \otimes \e^2 \\
		&\quad + \frac{A_3^2 - (A_1+A_2)^2}{2 A_1 A_2} \,\e^3 \otimes \e^3\,.
	\end{split}
\end{align}
The metric $\T h$ has an additional continuous symmetry if and only if $A_2 = A_3$, in which case the additional Killing vector is given by
\begin{align}\label{eq_KVHs_SL2R_fourth}
	\T\xi_4 = -\T\partial_x = \e_1\,,
\end{align}
and the metric is maximal.

\begin{proposition}
	The maximal number of linearly independent Ricci collineations of a maximal metric on $\SLRRs$-geometric manifolds is~$6$ and is achieved if and only if $A_1 = 2A_2$. In this case, the Ricci tensor is left and right invariant under $\SLRR$, and its collineations are
	\begin{equation}\label{eq_KVHs_SL2R_Ric}
		\T\xi_1\,,\ \T\xi_2\,,\ \T\xi_3\,,\ \e_1\,,\ \e_2\,,\ \e_3\,,
	\end{equation}
	with $\Sym(\tilde\Sigma, \sTRic)_0 = \SLRR \times \SLRR$.
\end{proposition}

\begin{remark}
From~\eqref{eq_app_SL2R_Ric_min}, we can see that a minimal metric with $A_1 = A_2 + A_3$ (equivalently with $n_1 + n_2 + n_3 = 0$ in an orthonormal Milnor basis) produces the same Ricci tensor as a maximal metric (i.e., $A_2 = A_3$, equivalently $n_2 = n_3$ in an orthonormal Milnor basis) with $A_1 = 2A_2$ (equivalently with $n_1 + 2n_2 = 0$ in an orthonormal Milnor basis).
\end{remark}

\begin{proof}
	The Ricci tensor of a maximal $\SLRRs$-metric is
\begin{equation}
	\sTRic[\T h] =
		\left(\frac{A_1}{2 A_2}\right)^2 \,\e^1 \otimes \e^1 + \frac{A_1(A_1+2A_2)}{2 A_2 A_1} \,\e^2 \otimes \e^2 + \frac{\left(A_2\right)^2 - (A_1+A_2)^2}{2 A_1 A_2} \,\e^3 \otimes \e^3\,.
\end{equation}
This tensor is always non-degenerate for any value of $A_1>0$ and $A_2>0$, with always two negative eigenvalues and one positive eigenvalue. Therefore, $\sTRic[\T h]$ is a Lorentzian metric and for this reason its maximum number of collineations is $6$. This is achieved if and only if $\sTRic[\T h]$ is an Einstein metric, i.e., $\sTRic\left[\sTRic[\T h]\right] \propto \sTRic[\T h]$. Solving the latter equation~yields
\begin{align}
	\sTRic\left[\sTRic[\T h]\right] = \frac{1}{4}  \sTRic[\T h],\quad {\rm iff} \ A_1 = 2 A_2\,.
\end{align}
This shows that for $A_1 = 2 A_2$, a maximal $\SLRRs$-metric has~$6$ linearly independent Ricci collineations. If $A_1 \not= 2 A_2$, then $\sTRic[\T h]$ is not an Einstein metric and the maximum number of Ricci collineations of $\T h$ is $4$. This is because, on a $3$-manifold, the isometry group of a metric cannot be 5-dimensional (cf., e.g., \cite[Section~1.2]{1997_Wainwright_et_al_BOOK}).

Finally, a direct computation of $\Lie{\e_i} \sTRic$ for all left-invariant vector basis (given in \Cref{tab_Bianchi_vectors}), shows that all $\e_i$ are collineations of $\sTRic$, and therefore  $\sTRic$ is right-invariant, implying $\Sym(\tilde\Sigma, \sTRic)_0 = \SLRR \times \SLRR$ (where the subscript ${}_0$ means the connected component of the identity). Interestingly, this implies that $\sTRic$ is the (opposite of the) Killing form of $\SLRR$ (e.g. \cite{2015_Olea}).
\end{proof}

\subsubsection{Bianchi VIII metrics as Bianchi III metrics}\label{app_SLRR_dico}

If a metric $\T h$ is at the same time a $\BVIII$ and a $\BIII$ metric, then it can be represented in two different orthonormal Milnor bases having the algebra of the former group and the latter group, respectively, hence non-isomorphic algebras. In the Bianchi VIII representation, $n^{\rm VIII}_2=n^{\rm VIII}_3$ is imposed by the presence of a fourth Killing vector additional to the three coming from the $\BVIII$ group. In the Bianchi III representation, $n^{\rm III}_2 \not = a^{\rm III}$ is imposed, otherwise the $\BIII$ metric is a $\GRHH$ metric. Since these are two different representations of the same metric, one can relate $n^{\rm VIII}_1$ and $n^{\rm VIII}_2$ to $n^{\rm III}_2$ and $a^{\rm III}$. To do so, it is sufficient to calculate the Ricci scalar $\tr_{\T h} \sTRic$ and the Ricci square $\sR_{ij} \sR^{ij}$. Since these are invariant quantities, they must agree in both bases. One finds
\begin{equation}
    n^{\rm VIII}_1
        = \frac{\left|\left(a^{\rm III}\right)^2 - \left(n^{\rm III}_2\right)^2\right|}{n^{\rm III}_2}\,, \\
    n^{\rm VIII}_2
        = \frac{4 \, \left(a^{\rm III}\right)^2\, n^{\rm III}_2}{\left|\left(a^{\rm III}\right)^2 - \left(n^{\rm III}_2\right)^2\right|}\,; \\ \forall\, n^{\rm III}_2 \neq a^{\rm III}\,.
\end{equation}

\subsection{$\RSS$-geometry}\label{app_RS2}

On $\RSS$, the generators of $\GRSS$-invariance are the vector fields
\begin{align}\label{eq_Killings_RS2}
    \begin{split}
    &\T\xi_1
        = \cos y\, \T\partial_x - \cot x\sin y \,\T\partial_y\,, \quad
    \T\xi_2
        = \sin y\, \T\partial_x + \cot x\cos y\, \T\partial_y\,, \\
    &\T\xi_3
        = \T\partial_y\,, \quad \T\xi_4 = \T\partial_z\,, 
    \end{split}
\end{align}
with
\begin{align}
	[\T\xi_1, \T\xi_2] = -\T\xi_3\,, \quad
	[\T\xi_1, \T\xi_3] = \T\xi_2\,, \quad
	[\T\xi_2, \T\xi_3] = -\T\xi_1\,,
\end{align}
and $\T\xi_4$ commutes with all the other vector fields.

A general $\GRSS$-invariant (Kantowski--Sachs) metric $\T h$ has the form
\begin{align}
	\T h
	    = A_1 \left(\T\dd x \otimes \T\dd x + \sin^2 x\, \T\dd y \otimes \T\dd y\right)
	        + A_2 \, \T\dd z \otimes \T\dd z\,.
\end{align}

By a direct computation with the Lie derivative along the vector fields~\eqref{eq_Killings_RS2}, one can prove the following proposition.
\begin{proposition}\label[poposition]{prop_property_RSS}
    Let $\Sigma$ be a closed manifold modeled on the $\RSS$-geometry. Let $\T v$, $\T n$ and $\T A$ be, respectively, a vector field, a 1-form and a symmetric $(0,2)$-tensor on $\Sigma$, all locally $\GRSS$-invariant. Then
	\begin{align*}
		\T v &= c_1\,\T\partial_z\,, \\
		\T n &= c_2\,\T\dd z\,, \\
		\T A &= c_3\, \left(\T\dd x\otimes\T\dd x + \sin^2 x\, \T\dd y\otimes\T\dd y\right)
		    + c_4 \, \T\dd z\otimes\T\dd z\,,
	\end{align*}
	where $c_i$ for $i=1,\ldots,4$ are constants.
\end{proposition}

A direct consequence of this proposition is
\begin{corollary}\label[corollary]{prop_property_RSS_bis}
    Given a locally $\GRSS$-invariant metric $\T h$ on an $\RSS$-geometric manifold~$\Sigma$, any symmetric $(0,2)$-tensor $\T A$ such that $\Isom(\tilde\Sigma,\T h) \subseteq \Sym(\tilde\Sigma,\T A)$, is a transverse tensor field, i.e., $\div_{\T h} \T A = 0$.
\end{corollary}

\newpage
\printbibliography[heading=bibintoc] 

\end{document}